\documentclass[12pt,fleqn]{article}
\usepackage{amsmath,amssymb,amsthm,amsfonts,mathrsfs}
\usepackage{bm}
\usepackage[a4paper]{geometry}
\usepackage[T1]{fontenc} 
\usepackage{lmodern}
\usepackage[latin1]{inputenc} 
\usepackage[english]{babel} 
\usepackage{url} 
\usepackage{indentfirst} 
\usepackage{graphicx} 
\usepackage{tikz-cd}
\graphicspath{{figures/}}
\DeclareGraphicsExtensions{.png,.pdf}
\usepackage{authblk}  
\newtheorem{assumption}{Assumption}

\newtheorem{remark}{Remark}

\newtheorem{theorem}{Theorem}

\usepackage{nicefrac}

\newcommand{\R}{\mathbb{R}} 
\newcommand{\D}{\mathrm{d}} 
\newcommand{\jp}{_{j + 1/2}}
\newcommand{\jm}{_{j - 1/2}}

\newcommand{\stl}{^{(\ell)}}

\newcommand{\dof}{{\em dof }\nolinebreak's}
\newcommand{\totder}[2]{\frac{\D}{\D #1}#2}
\newcommand{\transport}[1]{\partial_t #1 + v \cdot \nabla_x #1}

\newcommand{\G}{\mbox{G}_f}
\newcommand{\M}{\mbox{M}_f}
\newcommand{\Htheo}{\mathcal{H}}
\newcommand{\Mop}{\mathfrak M}
\newcommand{\disG}{{\mathcal G}_f}
\newcommand{\disM}{{\mathcal M}_f}
\newcommand{\Meq}{M_{\mbox{\scriptsize eq}}}
\newcommand{\Teq}{T_{\mbox{\scriptsize eq}}}
\newcommand{\Pra}{{\mbox{P}_{\mbox{r}}}}

\newcommand{\Dt}{\Delta t}
\newcommand{\Dv}{\Delta v}
\newcommand{\Dx}{\Delta x}
\newcommand{\ket}[1]{\large< #1 \large>}
\newcommand{\mom}{\mathbf{U}}

\title{Kinetic models of BGK type and their numerical integration}

\author[1]{Gabriella Puppo}

\affil[1]{La Sapienza Universit\`a di Roma, Roma, Italy\\

URL: {\tt gabriellapuppo.it} \\

email {\tt gabriella.puppo@uniroma1.it}}

\date{}

\begin{document}

\maketitle

\begin{abstract}
This minicourse contains a description of recent results on the modelling of rarefied gases in weakly non equilibrium regimes, and the numerical methods used to approximate the resulting equations. Therefore this work focuses on BGK type approximations, rather than on full Boltzmann models. Within this framework, models for polyatomic gases and for mixtures will be considered. We will also address numerical issues characteristic of the difficulties one encounters when integrating kinetic equations. In particular,  we will consider asymptotic preserving schemes, which are designed to approximate equilibrium solutions, without resolving the fast scales of the approach to equilibrium.
\end{abstract}

Kinetic theory was initially developed to study the behaviour of rarefied gases, with applications, as a typical example, to flow in the higher levels of the atmosphere. Recently its scope has enlarged to include many non equilibrium phenomena, arising, for instance, in the study of microfluids, i.e. flows occurring in domains with microscales, where the equilibrium hypothesis underlying classical gas dynamics does not hold. But kinetic models have also been successfully applied to phenomena which do not stem from fluid dynamics. The attractive feature of kinetic theory beyond gas dynamics is its ability to start from the characteristics of interactions of particles at a microscopic scale, to develop equations for the collective behaviour. These new applications include social sciences, see the examples in \cite{PareschiToscaniBook} or natural sciences \cite{Bellomo2013}, or even traffic flow, \cite{Visconti_TrafficBGK}. 
For a short introduction to rarefied gas dynamics see \cite{Cercignani2000}[Chapt. 1]. A more in depth text is \cite{Cercignani1988}

This paper will review kinetic models of BGK type, and the numerical techniques which permit to obtain accurate and reliable approximate solutions. The paper starts from the standard BGK model of \cite{BGK}, and continues to more recent developments, such as the  ES-BGK model, models for polyatomic gases and for mixtures of fluids composed of different particles.

The BGK model is an approximation of Boltzmann equation, but it contains the same approach to equilibrium. Moreover, numerical schemes for BGK are inherently faster than schemes for Boltzmann, because the equilibrium distribution is partly known. For this reason, BGK is also used as a numerical tool to accelerate the numerical solution of Boltzmann equation, see the penalization technique proposed in \cite{FilbetJin2010AP}, or as a buffer zone, connecting fully kinetic domains to equilibrium regions, in domain decomposition strategies, \cite{Alaia:2012ex}.

From the general Boltzmann equation for a mixture of gases, we will consider the simplified BGK model, which provides a good approximation of the Boltzmann equation close to equilibrium. The main properties of the BGK model, and its improved version, the ES-BGK model, are studied in \cite{AndriesPerthameESBGK2000}. I will summarize those results, which are relevant for a single gas.

Next, for mixtures of gases, two main approaches appear in the literature. One was introduced by Aoki et al. in \cite{AndriesAokiPerthame2002}. It is characterized by a single collision kernel for each distribution function, and the purpose is to be consistent with the single species model, in case the mixture is composed of identical gases. 

Another approach reflects the structure of Boltzmann models for mixtures, and is characterized by a collision kernel for each type of interaction. It is the classical choice in engineering applications. In \cite{KlingPirner:2017} we have generalized existing models of this kind, and proposed a unified approach, which permits to study the main properties of the model, such as the relaxation towards equilibrium. This is the model I will discuss more in detail. The model has been further improved in \cite{BobylevGroppi2018}, where the parameters driving momentum and energy exchanges due to interspecies collisions have been computed, starting from the full Boltzmann collision terms. 

The numerical integration of kinetic equations presents several challenges. The first difficulty is due to the large number of independent variables, because, beside space and time, one must include also the microscopic speeds in the number of degrees of freedom. Thus it is mandatory to use coarse grids, whenever possible. Coarse grids however may lead to large errors in the macroscopic variables. To control these errors, imposing conservation at the discrete level, one must use discrete Maxwellians. This issue was settled in \cite{Mieussens2000Discrete}. Further, to control the error while using coarse grids, high order methods should be considered.

The most challenging difficulty however is due to the stiffness of the collision term, when the system is close to equilibrium. To prevent the use of very small time steps, it is necessary to use implicit or semi-implicit schemes. The problem is that in doing so it is important to preserve the correct equilibrium solutions. This leads to the notion of Asymptotic Preserving schemes (AP). Some recent reviews can be found in \cite{DiMarcoPareschiActa} and \cite{JinLi2017}. Here I will start from the AP method of \cite{PieracciniPuppo2007}, which can be easily extended to mixtures, continuing with methods proposed in \cite{DiMarcoPareschi2013}.

Many interesting results which have been obtained on these models cannot be included in this work. Thus,  I wish to widen the perspective a bit, citing some papers which deal with issues that I will not touch.

Several models for mixtures can be found in the literature. Many provide extensions of the single collision model of \cite{AndriesAokiPerthame2002}, to include chemical reactions, \cite{Bisi2010}, or \cite{GroppiSpiga2013reacting}, ES-BGK extensions \cite{GroppiSpiga2011ES, Brull2014ESmixtures}, polyatomic effects \cite{Bisi2016PolyMixtures}. Models based on a collision term for each type of interaction are more frequent in engineering literature, as \cite{Garzo1989}. Extensions to the polyatomic case are proposed in \cite{KlingPirnerPuppo2018}, or in \cite{BarangerBisiBrull2018PolyMixtures}.

The existence, uniqueness  and positivity of the solution of the mixture model based on a collision kernel for each type of interaction have been proven in \cite{KlingPirner2018}. For a recent review on the important results about the macroscopic limits of kinetic equations see \cite{Golse2014FluidLimits} and its references.

The literature on numerical methods for kinetic problems is very crowded. Here I will concentrate on efficient AP schemes. The first work I know which used a blend of implicit and explicit schemes for kinetic problems is \cite{PareschiRusso2005}. The construction of AP schemes for kinetic problems is described in \cite{FilbetJin2010AP}. This was extended to the ES-BGK case in \cite{FilbetJin2010ES}. An interesting approach to the numerical integration of kinetic equations, based on deviations from equilibrium can be found in \cite{BennouneLemou2008MicroMacro}. This work was followed by several other applications of analogous techniques, for instance \cite{CrestettoCroueilles2012Micro}. A similar approach has also been applied to mixtures of gases, see \cite{JinShiMicroMacro}, \cite{CrestettoKlingMixtures}, or \cite{JinLi2013BGKMixtures}. 

Unfortunately, this paper does not contain numerical results. I wished to address at the same time the richness of these models, and of the corresponding numerical techniques. A discussion and a comparison of numerical results would have required much more space, and they can be found in the references included in this work.


\tableofcontents
\section{BGK: where does it come from?}\label{s:BGK}
The BGK model was proposed in \cite{BGK}. A formal derivation can be described as follows. One starts from the standard Boltzmann equation, for the distribution function $f(x,v,t)$, where $f/n(x,t)$ is the probability density of finding a particle at the point $(x,v,t)$ in phase space, and  $n(x,t)$ is the number of particles per unit volume,
\begin{equation}\label{e:Boltz}
\transport{f} =  \int_{\R^d}\int_{S^+}B(|v-v_*|,\hat{n})\left[f(v')f(v'_*) - f(v)f(v_*)\right]\; \D v_*\; \D \omega.
\end{equation}
Here  $t$ is time, $x\in \R^d$ is space, $v_*$ is the velocity of a field particle interacting with the test particle with velocity $v$, while $v'$, $v'_*$ are the post interaction velocities, and the dependence on $x$ and $t$ of the distribution function has been dropped for shortness. Further  $\hat{n}$ is the normal to the vector $v-v_*$, $B(|v-v_*|,\hat{n})$ is the collision cross section of the interaction, and the integral over $S^+$ is the integral over all incoming directions towards the field particle.

The right hand side of the Boltzmann equation \eqref{e:Boltz} is called {\em collision term}. It is convenient to introduce the short hand
\[
Q^B(f,f) = \int_{\R^d}\int_{S^+}B(|v-v_*|,\hat{n})\left[f(v')f(v'_*) - f(v)f(v_*)\right]\; \D v_*\; \D \omega.
\]
The collisions conserve mass, momentum and energy, which, at the microscopic level,  means
\begin{align*}
& v'+v'_* = v+v_* \\
& (v')^2+ (v'_*)^2 = (v)^2+(v_*)^2.
\end{align*}
This implies that in each collision the total masses, momentum and energies of the particles involved in the collision is preserved. Summing over all possible collisions, it is easy to see that
\[
\int_{\R^d} Q^B \D v=0, \qquad \int_{\R^d}v\, Q^B \D v=0 \qquad \mbox{and}\quad \int_{\R^d}v^2\, Q^B \D v=0,
\]
see \cite{Cercignani2000}. The quantities $\phi(v)=[1, v, \tfrac12 v^2] \in \R^{d+2}$ are called {\em collision invariants}. Thus, the conservation properties of the collisions at the microscopic level imply global conservation, which can be expressed as
\begin{equation}\label{e:cons}
\int_{\R^d} \phi(v)\, Q^B \D v=0.
\end{equation}

At the macroscopic level, mass, momentum and energy are obtained as expected values or  {\em moments} of the distribution function $f$ in velocity space. Let $m$ denote the particle mass, then
\begin{align}
\rho(x,t) & = m \int_{\R^d} f(x,v,t) \; \D v  \label{e:mass}\\
(\rho u) (x,t)&  = m\int_{\R^d} v\, f(x,v,t) \; \D v \label{e:momentum} \\
E(x,t) & = m \int_{\R^d} \tfrac12 v^2\, f(x,v,t) \; \D v, \label{e:energy} 
\end{align}
where $\rho$ is the density, $u$ is the macroscopic velocity and $E$ is the energy density per unit volume. Thanks to \eqref{e:cons}, multiplying the Boltzmann equation by the collision invariants and integrating in velocity space, one obtains conservation laws for mass, momentum and energy.

The system is in equilibrium when the collision term is zero. Thus the equilibrium distribution functions are given by solutions of the integral equation $Q^B(f,f)=0$. It is easy to prove, see again any textbook on kinetic theory as \cite{Cercignani2000}, that the conservation properties of the collisions dictate the structure of the equilibrium solutions which are called 
{\em Maxwellian distributions}, and are given by the expression
\begin{equation}\label{e:Maxwellian}
\M = \frac{n(x,t)}{(2\pi KT^0(x,t)/m)^{d/2}}\; {\Large e^{-\frac{m(v-u(x,t))^2}{2KT^0(x,t)}}},
\end{equation}
where  $K$ is the Boltzmann constant, $K=1.38\; 10^{-23} \mbox{Joule}/K^0$ and $T^0(x,t)$ is the temperature measured in  degrees Kelvin, while $m$ is the mass of the molecules composing the gas, which means that, for the moment, we are considering a gas composed of a identical particles. Often, the Maxwellian is written in terms of the {\em specific} gas constant $R$, which is defined as $R=K/m$. The specific gas constant is obtained from the universal constant of gas $\overline{R}$ as $R=\overline{R}/M$, where $M=mN_A$ is the weight of a mole of gas, $N_A$ being Avogadro's number. In the following, I will always write $T$ intending $KT^0$. Therefore, $T$ will have the physical dimensions of an energy, and $\tfrac{d}{2}T$ can be thought of as the average energy of a molecule, in the gas frame of reference. When $f=\M$, the right hand side of eq. \eqref{e:Boltz} is zero, thanks to the conservation properties of the microscopic collisions. Thanks to the $\Htheo$-theorem (see below), it is possible to show that the Maxwellian is the unique equilibrium distribution of the gas, and that an isolated gas approaches equilibrium as it evolves.

\subsection{The standard BGK model}

In the BGK approximation, one supposes that
\begin{itemize}
\item The cross section of the interaction does not depend on the relative velocity of the interacting particles, $B(|v-v_*|,\hat{n})=B(\hat{n})$.
\item The particles reach equilibrium after a single interaction, so the distribution of post interaction particles is Maxwellian.
\end{itemize}

Since the cross section does not depend on $v_*$, the loss term in eq. \eqref{e:Boltz} can be rewritten as
\[
\begin{split}
 \int_{\R^d}\int_{S^+}B(|v-v_*|,\hat{n}) f(v)f(v_*) \D v_*\; \D \omega & = f(v) \int_{S^+}B(\hat{n})\D \omega \int_{\R^d} f(v_*) \D v_* \\
& = \nu(x,t) n(x,t) \; f(x,t,v). 
\end{split}
\]
where $\nu=\int_{S^+}B(\hat{n})$. For the gain term instead, postcollisional distributions are Maxwellians, then
\[
\begin{split}
 & \int_{\R^d}\int_{S^+}B(|v-v_*|,\hat{n}) f(v')f(v_*') \D v_*\; \D \omega \\ & =  \int_{\R^d}\int_{S^+}B(|v-v_*|,\hat{n}) \M(v')\M(v_*') \D v_*\; \D \omega  \\
& =  \int_{\R^d}\int_{S^+}B(|v-v_*|,\hat{n}) \M(v)\M(v_*) \D v_*\; \D \omega  \\
& =  \M(v) \int_{S^+}B(\hat{n}) \D \omega \int_{\R^d}\M(v_*) \D v_*\;   \\
& = \nu(x,t) n(x,t) \;\M(x,t,v),
\end{split}
\]
where I have used the conservation properties of the microscopic collisions to rewrite the equation in terms of the {\em precollisional } Maxwellians.

Substituting these approximations of the collision term into the Boltzmann equation \eqref{e:Boltz}, we obtain the BGK model
\begin{equation}\label{e:BGK}
\transport{f} = \nu\, n (\M -f).
\end{equation}
 Note that $\nu n$ has the physical dimensions of $\mbox{time}^{-1}$. I will often write $\tau=1/(\nu n)$ which is the relaxation time. A typical model for $\tau$ \cite{chapman1970mathematical} is given by
\begin{equation}\label{e:tau}
\frac{1}{\tau} = \frac{nT}{\mu_0}\left( \frac{T_0}{T} \right)^{\omega} \qquad \Longrightarrow \qquad \nu= \frac{T}{\mu_0}\left( \frac{T_0}{T} \right)^{\omega}
\end{equation}
where $\mu_0$ and $T_0$ are the reference viscosity and temperature of the gas and $\omega$ is the exponent of the viscosity law of the gas. I remind that $T$ in these notes has the dimensions of an energy, while the viscosity $\mu_0$ has dimensions of mass divided by length per time.

The equilibrium solution for BGK is clearly $\M=f$, as in Boltzmann equation. Moreover, BGK has the same collision invariants $\phi(v)=[1, v, \tfrac12 v^2]$ of Boltzmann: namely, multiplying  \eqref{e:BGK} by $m\phi(v)$ and integrating in velocity space, one gets the set of conservation laws
\begin{align}
& \partial_t \rho + \nabla \cdot (\rho u) = 0 \label{e:MassCons}\\
& \partial_t (\rho u) + \nabla \cdot (\rho u \otimes u + \mathbb{P}) = 0 \label{e:MomCons}\\
& \partial_t E + \nabla \cdot ( u E + \mathbb{P}u + q) = 0, \label{e:ECons} 
\end{align}
where
\begin{equation}\label{e:pTens}
\mathbb{P} = m\int_{\mathbb{R}^d} (v-u)\otimes(v-u)\, f\, \D v
\end{equation}
is the pressure tensor, and 
\begin{equation}\label{e:qFlux}
q = \frac{m}{2}\int_{\mathbb{R}^d} (v-u)||(v-u)||^2\, f\, \D v
\end{equation}
is the heat flux. The system of equations \eqref{e:MassCons}, \eqref{e:MomCons} and 
\eqref{e:ECons} is not a closed system of equations, because we don't know the relation between $\mathbb{P}, q$ and the conservative variables in $\rho, \rho u, E$. To solve the system, we need $f$, from which we can compute the stress tensor and the heat flux.

At equilibrium, the pressure tensor and the heat flux can be computed explicitly, substituting $f=\M$ in \eqref{e:pTens} and \eqref{e:qFlux}. In fact, the $v$ dependence of the Maxwellian is completely known, and therefore the integrals in velocity space which yield the moments can be computed explicitely. By symmetry, one obtains $q=0$, and $\mathbb{P}$ diagonal. Let $p=\frac{1}{d}\mbox{Tr}\, (\mathbb{P})$ be the pressure, where $d$ is, as already indicated, the dimensions of the velocity space. The energy can be split into kinetic and internal energy as
\[
E = \frac12 \rho u^2 + m \frac12 \int_{\R^d} ||v-u||^2 f \D v= \frac12 \rho u^2 + \rho\, e.
\] 
The link between pressure and internal energy is $\rho e= \tfrac{d}{2}p$. Usually, we can suppose that each degree of freedom in velocity space corresponds to the same energy (equipartition of energy), and we define the temperature as
\begin{equation}
 T = \frac{1}{n}m \int_{\R^d} (v_i-u_i)^2 f \D v= \frac{2}{dn}m \int_{\R^d}\frac12 (v-u)^2 f \D v,
\end{equation}
which means that $dT/2$ is the average kinetic energy of a molecule, in the gas frame of reference. Thus, the average energy of a molecule is $\tfrac12 mu^2+ \tfrac{d}{2}T$.
With this definition, the internal energy per unit mass $e$ and the temperature are linked by
\[
e = \frac{d}{2m}T.
\]
From this relation and the link between the energy and the pressure, one finds the equation of state of the gas
\[
p= \frac{1}{m}\rho T. 
\]
Consider a volume $V$ containing $n_m$ moles of the gas. Then the volume $V$ contains $n=n_m N_A$ molecules, where $N_A$ is Avogadro's number, and the density can be written as $\rho=m n_m N_A/V$. The expression for the pressure becomes
\begin{equation}
pV = n_m N_A T= n_m N_A K T^0 = n_m \overline{R}T^0,
\end{equation}
which is the familiar equation of state of a perfect gas.

Substituting the equilibrium distribution $f=\M$ in \eqref{e:MassCons}, \eqref{e:MomCons} and \eqref{e:ECons}, thanks to the equation of state, one finds
Euler system of classical gas dynamics
\begin{align}
& \partial_t \rho + \nabla \cdot (\rho u) = 0 \label{e:Euler}\\
& \partial_t (\rho u) + \nabla \cdot (\rho u \otimes u + p) = 0 \nonumber\\
& \partial_t E + \nabla \cdot ( u (E + p)) = 0, \nonumber 
\end{align}
closed by the equation of state $p=\frac{2}{d}\rho e$. For monoatomic gases, each molecule has only $d=3$ degrees of freedom, thus $p=\frac{2}{3}\rho e$. For bi-atomic molecules, one must also consider the 2 degrees of freedom given by the possible rotation of the molecule around its axis. In this case, $p=\frac{2}{5}\rho e$, which is the standard equation of state used for air, namely, $p=\rho e(\gamma-1)$, with $\gamma=1.4$. Polyatomic gases will be discussed more in detail in \S \ref{s:Poly}.

\subsubsection{The $\Htheo$-theorem}

A very important property of kinetic models is the $\Htheo$-theorem, which states that entropy decays as the system evolves, until equilibrium is reached.

To illustrate what are the ingredients that draw the gas towards equilibrium, I will include a proof of the $\Htheo$-theorem for the BGK model in the space homogeneous case. More generally, it will be enough to consider the {\em total} entropy, integrating also in space, and assuming suitable decays properties of $f$ at the boundary, see \cite{Cercignani1988}. First, we define the entropy, which is the quantity
\begin{equation}\label{e:entropy}
H(f)(x,t)= \int_{\R^d} f\, \ln f\; \D v.
\end{equation}

\begin{theorem}{\bf $\Htheo$-theorem for the BGK model}\label{H-theorem}
Consider the BGK model in the space homogeneous case. Then the entropy $H$ decays in time,
\[
\frac{\D H(f)}{\D t}  \leq 0,
\] 
with equality if and only if $f=\M$, where $\M$ is the Maxwellian defined by the moments of $f$.
\end{theorem}

\begin{proof}
The proof is very standard, see for instance \cite{Cercignani1988}. I will give a sketch, to show where the main ingredients play a role.
Consider the identity:
\[
(1+\ln f) \partial_t f = \partial_t \big( f\, \ln f \big).
\]
Multiply the space homogeneous BGK model \eqref{e:BGK} by $(1+\ln f)$ and integrate in velocity space. You find,
\begin{equation}\label{e:Htheo}
\begin{split}
\partial_t H & = \frac{1}{\tau} \left[\int_{\R^d} \ln f (\M-f) + \int_{\R^d} (\M-f). \right]\\
    & = \left[\int_{\R^d} \ln f (\M-f) - \int_{\R^d} \ln \M (\M-f) \right] \\
    & = - \int_{\R^d} \frac{\ln f - \ln \M}{f-\M}(f-\M)^2\leq 0,
 \end{split}
\end{equation}
where the conservation laws for mass momentum and energy have been used at the second step and  the convexity of the logarithm at the final step. 
\end{proof}

Note that the proof of the $\Htheo$ theorem depends strongly on convexity, and on all conservation principles. Another way to prove the $\Htheo$ is to prove that the Maxwellian is the unique minimizer of the constrained minimization problem
\begin{equation}\label{e:minimo}
\begin{split}
H(\M) &= \min_{f\in \mathcal{S}(n,u,T)} H(f), \\
      & {\small \mathcal{S}(n,u,T) = \left\{f\geq 0, \;\mbox{s.t.}\; \int\!\! f=n;\; \int\!\! vf = nu; \; \int \!\!(v-u)^2f = \frac{n\,d}{m}T \right\} }.
\end{split}
\end{equation}

\subsection{The ES-BGK model}

When the gas is in equilibrium, the distribution function $f$ coincides with the Maxwellian $\M$ with the same moments of $f$. Thus, BGK describes correctly the equilibrium, because when $f=\M$ the collision term vanishes, and, at the macroscopic level, one recovers compressible gas dynamics. On the other hand, when the gas is weakly off equilibrium, a standard procedure is to apply a Chapman-Enskog expansion \cite{chapman1970mathematical}, in which $f=\M + \tau g + o(\tau)$, where $g$ has zero moments. Keeping terms of order $\tau$ and disregarding higher order terms, at the macroscopic level, the Compressible Navier-Stokes (CNS) equations are obtained,
\begin{align}
& \partial_t \rho + \nabla \cdot (\rho u) = 0 \label{e:CNS}\\
& \partial_t (\rho u) + \nabla \cdot (\rho u \otimes u + p) = \nabla \cdot \big(\mu \sigma \big) \nonumber\\
& \partial_t E + \nabla \cdot ( u (E + p)) = \nabla \cdot \big( \mu \sigma u + \kappa q\big), \nonumber 
\end{align}
where $\mu$ and $\kappa$ are, respectively, the viscosity and the heat conductivity of the gas, and are $O(\tau)$, while
\[
\sigma_{i,j} = \partial_{x_i} u_j + \partial_{x_j} u_i - \frac{2}{d}\nabla \cdot u\, \delta_{i,j}, \qquad \mbox{and}\quad q_i = \nabla T
\]
are the tensor of viscous stresses and the heat flux.
Thus, viscous terms and heat fluxes are derived as off-equilibrium phenomena.

The main difficulty of the standard BGK model is the presence of a single adjustable parameter, $\tau$, which permits to adjust the model to only one of the two terms that appear in the Chapman Enskog expansion, namely, either viscosity or heat fluxes. The ratio between viscous and heat fluxes is measured by the adimensional parameter $\Pra$, the Prandtl number, which is given by
\begin{equation}
\Pra = \frac{\gamma}{\gamma-1}\frac{R\mu}{\kappa}, \qquad \mbox{with}\; \gamma=\frac{d+2}{d}
\end{equation}
where $\mu$ and $\kappa$ are, as we have already seen, the viscosity and heat conductivity obtained with the Chapman-Enskog expansion, which gives their dependency on $T$, while $d$ is the number of microscopic velocity degrees of freedom. As discussed in \cite{Andries:2001vi}, the standard BGK model gives $\Pra=1$, thus only one of the two coefficients, $\mu$ and $\kappa$, can be matched adjusting $\tau$. Instead, the correct value of the Prandtl number for a monoatomic gas is $\Pr\simeq \tfrac23$.

To overcome this difficulty, several alternatives have been proposed. Here, I will discuss the Ellipsoidal Statistic (ES-BGK) model, proposed by Holway in \cite{Holway-ESBGK}. I will follow the version found in \cite{Andries:2001vi}, where the well posedness of the model was finally proven (positivity of the distribution function and $\Htheo$ theorem), which enabled the ES-BGK model to be considered as a sound model from a mathematical point of view.

The idea is to construct a BGK like operator, where the Maxwellian is substituted with a  distribution function which becomes Maxwellian only at equilibrium.
Introduce the normalized stress tensor, $\Theta$ defined  by
\begin{equation}
n\Theta = m \int_{\R^d} (v-u)\otimes(v-u) \, f\; \D v.
\end{equation}
Then, to build the attractive distribution, we consider a combination of the temperature and of the stress tensor. Namely, let
\begin{equation}\label{e:TensorESBGK}
\mathbb{T} = (1-\omega)T \, \mathbb{I}+ \omega \Theta,
\end{equation}
where $\mathbb{I}$ denotes the $d\times d$ identity matrix. The ES-BGK model is defined by
\begin{align}\label{e:ES-BGK}
&\transport{f} = \frac1\tau \Big( \G - f\Big) \\
& \G = \frac{n}{\sqrt{ \det \left(2\pi \frac{\mathbb{T}}{m}\right)}}\; \mbox{exp}\Big( -\frac{m}{2}(v-u)^T\mathbb{T}^{-1}(v-u) \Big).\nonumber
\end{align}
The matrix $\mathbb{T}$ is the ``temperature'' tensor, and $\G$ is the  distribution to which $f$ relaxes to. Since, as $\tau \to 0$, we expect the gas to approach equilibrium, the stress tensor should become diagonal, and, because of equipartition of energy, the diagonal terms of $\Theta$ should all be equal, and coincide with the temperature. From \eqref{e:ES-BGK}, we see that as $\tau\to 0$, $f\to \G$. What is remarkable is that it is possible to prove that while $f\to \G$,  $\G \to \M$, so that the ES-BGK model has the same equilibrium of the standard BGK model. This intuitive idea can be formalized.

\begin{theorem}
 The distribution function $\G$ has the same macroscopic moments of the function $f$. Moreover, suppose that $-\tfrac12 <\omega<1$, then, for $\tau \to 0$, in the space homogeneous case, the system reaches the Maxwellian equilibrium $\M$.
\end{theorem}
\begin{proof}
This proof is sketched from \cite{AlaiaThesis}, and \cite{Andries:2001vi}. By construction, the tensor $\mathbb{T}$ defined in \eqref{e:TensorESBGK} is symmetric. Further, in \cite{Andries:2001vi}, it is proven that it is also positive definite, provided $-\tfrac12 <\omega<1$. This is not trivial, because the combination in \eqref{e:TensorESBGK} is not convex for the values of $\omega$ of interest. Thus one can introduce the non singular matrix $L$ such that $\mathbb{T}^{-1}=L^TL$. This enables to define the change of variables $\xi = \sqrt{\tfrac{m}{2}}L (v-u)$, with Jacobian $J$ and determinant $\mbox{det}\,J=\sqrt{\mbox{det}(\frac{m}{2}L^TL)}$. Thus the expected value of $\G$ is
\[
\int_{\R^d} \G\; \D v = \frac{n}{\sqrt{\pi}^d}\int_{\R^d} e^{-\xi^T\xi}\; \D \xi = n = \int_{\R^d} f\; \D v.
\]
An analogous approach leads us to 
\begin{align*}
\int_{\R^d} v\G\; \D v & = \int_{\R^d} v f\; \D v = nu\\
m\int_{\R^d} \tfrac12|| v||^2\G\; \D v & = m\int_{\R^d}\tfrac12 || v||^2 f\; \D v = \tfrac12 \rho ||u||^2 + \tfrac12 n\,d\,T.
\end{align*}
Next, we observe that the stress tensor $\Theta$ satisfies the equation
\[
\partial_t (n\Theta) = \frac{1}{\tau}\left(m \int (v-u)\otimes(v-u) \G\; \D v-n\Theta\right).
\]
The integral involving $\G$ can be explicitly evaluated to give
\[
m \int (v-u)^T\G(v-u) \; \D v = \frac{2n}{\sqrt{\pi^d}} L^{-1} \left( \int_{\R^d} \xi \, \xi^T e^{-\xi^T\xi}\right) L^{-T} =n\, \mathbb{T}.
\]
Substituting the expression for $\mathbb{T}$ and eliminating $n$ which is constant in time, one obtains the time evolution for the stress stensor,
\begin{equation}\label{e:StressDecay}
\partial_t \Theta= \frac{1-\omega}{\tau}\Big(T\,\mathbb{I}-\Theta\Big).
\end{equation}
Thus, we have a relaxation law for the stress tensor, provided $\omega<1$, towards the diagonal matrix $T\,\mathbb{I}$, with entries equal to $T$ on the main diagonal. Substituting this information in \eqref{e:TensorESBGK}, we see that $\mathbb{T}\to T\mathbb{I}$ as $\tau\to 0$, which means that $\G\to \M$.
\end{proof}

The parameters $\omega$ and $\tau$ can be adjusted to reproduce the correct Prandtl number of the gas. In fact a Chapman Enskog expansion, applied to the ES-BGK model,  leads to
\[
\Pr = \frac{1}{1-\omega},
\]
so that, for $\omega=-\tfrac12$, the correct value of the Prandtl number is recovered. Thus, it is important to be able to choose a non convex combination in 
the definition of $\mathbb{T}$ in \eqref{e:TensorESBGK}, because the desired result is obtained with a negative value of $\omega$.

One of the main results of \cite{Andries:2001vi} is the first proof of the $\Htheo$   theorem for the ES-BGK model. The proof uses a constrained minimization problem similar to \eqref{e:minimo}, but in which the space of the constraint is generalized to account also for the ``temperature'' tensor $\mathbb{T}$. With this proof, the ES-BGK model has been freed to be used extensively in computations, because it reproduces the correct viscous and heat exchanges, but has also a sound mathematical background: positivity of the distribution function and entropy decay.

\section{BGK models for polyatomic gases}

In polyatomic gases, each particle has an energy that depends not only on traslational degrees of freedom, but also on rotations and vibrations. The traslational degrees of freedom describe the motion of the molecule across space and they are responsible for the kinetic energy of the molecule. The remaining modes are internal, and they do not result in the bulk movement of the gas, but they do contribute to the temperature and to the heat exchanges within the gas and with the external environment.

For the sake of simplicity, in this section I will consider biatomic molecules, but the models can be easily extended to more complex particles. In many applications, air can be considered as a bi-atomic gas. Experimental data, see for instance \cite{Semyonov1984} and DSMC (Direct Simulation Monte Carlo) simulations \cite{Frezzotti1997} show that the heat transfer and the shock structure can be very different for bi-atomic and mono-atomic fluids. Thus, an extension of the BGK model to treat also the polyatomic case is clearly very important in applications.

Several researchers have proposed kinetic models within the BGK approach to account for the complexity of polyatomic gases. Here I will discuss the new model proposed in \cite{Bernard2018Poly}, and the one analyzed in \cite{Andries:2001vi}, but see also \cite{Bisi2016PolyMixtures} or \cite{Brull2009Poly}.

\subsection{A multi-temperature BGK model}

Let $d$ denote the dimensions of physical space and $r$ the number of internal degrees of freedom, which, for a biatomic molecule, are the two rotational degrees of freedom, perpendicular to the axis of the molecule. Let $D=d+r$ be the total number of degrees of freedom. Then the microscopic velocity is a point $v\in \R^D$, $\tfrac12 m (v_1^2+\dots +v_d^2)$ is the kinetic energy of the molecule, while $\tfrac12 m (v_{d+1}^2+\dots +v_{d+r}^2)$ is the internal energy due to the particle rotation.

The macroscopic quantities are given by
\begin{align}
\rho & = m\int_{\R^D} f\; \D v \\
u_k & = \frac{m}{\rho} \int_{\R^D} v_k f\; \D v \quad k=1,\dots, d. \nonumber \\
0 & = \frac{m}{\rho} \int_{\R^D} v_k f\; \D v \quad k=d+1,\dots, D. \nonumber
\end{align}
Here we have considered the fact that the traslational velocities may result in a bulk movement of the gas with speed $u\in \R^d$, while the components of the microscopic velocity corresponding to the rotational degrees of freedom have an expected value equal to zero.

In standard gas dynamics, one supposes that each degree of freedom contributes equally to the energy of the gas, so to each degree of freedom one associates the same temperature. This however is not true for polyatomic gases, because the rotational and the traslational temperatures decay towards equilibrium with different relaxation rates. Thus we can say that the gas has two internal energies per unit mass $e_t$ and $e_r$ given by

\begin{align}
\rho e_t & =  \tfrac{1}{2} m \int_{\R^D} \sum_{k=1}^d (v_k-u_k)^2 f\, \D v \label{e:Etrans}\\
\rho e_r & =  \tfrac{1}{2} m \int_{\R^D} \sum_{k=d+1}^D (v_k)^2 f\, \D v \label{e:Erot}.
\end{align}

We assume that similar degrees of freedom have the same temperature, in a sort of partial equipartition of energy, so to $e_t$ and $e_r$ we associate two temperatures, namely $T_t$, the traslational temperature, and $T_r$, the rotational temperature, which are given by
\begin{equation}\label{e:PolyT}
T_t= \frac{2m}{d}e_t \qquad \mbox{and}\qquad T_r= \frac{2m}{r}e_r.
\end{equation}

The challenge is to model a system in which the energy decays are different for the traslational and the rotational modes. In the following, I will describe the model we proposed in \cite{Bernard2018Poly}.

The first equation is a standard BGK like relaxation, but the Maxwellian is characterised by two yet unknown temperatures,
\begin{equation}
\partial_t f+v\cdot\nabla_x f= 
\dfrac{1}{\tau}\Big(\M-f\Big)  \label{eq:BGK_poly} 
\end{equation}
\begin{multline}
\M(x,v,t)= n (x,t)\prod_{k=1,d}\Big(\frac{m}{2\pi \Lambda_t}\Big)^{1/2}\mbox{exp}\Big(-\frac{m}{2\Lambda_t}(v_k-u_k)^{2}\Big) \; \times\\
 \prod_{k=d+1,D}\Big(\frac{m}{2\pi \Lambda_r}\Big)^{1/2}\mbox{exp}\Big(-\frac{m}{2\Lambda_r}(v_k)^{2}\Big). \nonumber
\end{multline}
The temperatures $\Lambda_t$ and $\Lambda_r$ are non equilibrium temperatures which eventually will decay to a common temperature $\Teq$. The relaxation towards the equilibrium temperature is  governed by energy conservation. This equation is obtained by the relaxation of the local Maxwellian $\M$ to the equilibrium Maxwellian $\Meq$,
\begin{align}
&\partial_t \M+v \cdot\nabla_x \M= \dfrac{1}{Z_r\tau}\Big(\Meq-\M\Big)  \label{eq:BGK_polyMax} \\
& \Meq(x,v,t)=\dfrac{n(x,t)}{(2\pi \Teq(x,t)/m)^{D/2}}\mbox{exp}\Big(-\dfrac{m||v-\tilde{u}||^{2}}{2\Teq(x,t)}\Big). \nonumber
\end{align}
where $\tilde{u}=[u_1,\dots,u_d,0,\dots,0]^T\in \R^D$.
Here, $Z_r$ is a parameter that accounts for the fact that the rotational collision frequency is a priori different from the traslational collision frequency, thus the relaxation time towards a common temperature $\Teq$ is governed by a  characteristic time $Z_r\tau$ which can be different from the relaxation time $\tau$ appearing in the evolution of $f$.

Since the Maxwellian is a known function of $v$, and $\M$ and $\Meq$ share the first moments, namely $n(x,t)$ and $u(x,t)$, the only quantities that need to be found are the two partial temperatures $\Lambda_t$ and $\Lambda_r$. Thus, the relaxation equation for the Maxwellian can be reduced, multiplying it by  $m\prod_{k=d+1,D}v_k^2$ and integrating in phase space, to yield the evolution equation for the rotational energy,
\begin{equation}\label{eq:ConsTrot}
\partial_t(\rho \Lambda_r)+\nabla_x(\rho u\Lambda_r)=\dfrac{\rho}{Z_r\tau}(\Teq-\Lambda_r).
\end{equation}
This equation regulates the heat exchange between the different degrees of freedom.
which can be simplified using mass conservation to give:
\begin{equation}
\partial_t \Lambda_r+u\cdot\nabla_x \Lambda_r=\dfrac{1}{Z_r\tau}(\Teq-\Lambda_r).
\label{eq:Tr_closure}
\end{equation}
The system is then closed imposing that total energy is conserved in \eqref{eq:BGK_poly}:
\begin{align*}
0 & = m\int_{\R^{\D}}\dfrac{1}{2}\sum_{k=1}^{\D}v_k^2 \Big( \M -f \Big) \\
&=\dfrac{1}{2}\, n\, \Big( d\Lambda_t+ r \Lambda_r - d\,T_t- r\,T_r \Big)
\end{align*}
which simply says that
\begin{equation}\label{e:PartialTemperatures}
d\Lambda_t+ r \Lambda_r = d\,T_t + r\,T_r. 
\end{equation}
Applying conservation of energy also to \eqref{eq:BGK_polyMax}, one obtains the second closure relation needed,
\begin{equation}\label{e:EqTemperature}
d\Lambda_t+ r \Lambda_r = (d + r) \Teq. 
\end{equation}
The whole model is composed by \eqref{eq:BGK_poly}, the {\em scalar} equation \eqref{eq:Tr_closure} or its conservative version \eqref{eq:ConsTrot} and the energy conservation constraints \eqref{e:EqTemperature}, \eqref{e:PartialTemperatures}. As an example $Z_r$, which typically is larger than 1, can be chosen as in \cite{PolyCoeff}.

\subsubsection{$\Htheo$ theorem}

The polyatomic model just proposed is well posed because it satisfies an $\Htheo$ theorem, and, at least in the space homogeneous case, the distribution function remains positive for all times, if the initial data are non negative.

Let us define the entropy for the polyatomic model as
\begin{equation}
H(f)=\int_{\R^{\D}}f \ln f \D v +
  Z_r \int_{\R^{\D}}\M \ln \M \D v.
\label{eq:H-function}
\end{equation}

\begin{theorem}{\bf $\Htheo$-theorem for the polyatomic BGK model}\label{H-theoremPoly}

\noindent Let $f$ be the solution of the polyatomic model \eqref{eq:BGK_poly}, and let $Z_r\geq 1$. Suppose that at the initial time the distribution function $f$ is non negative, then, in the space homogeneous case,
\begin{displaymath}
\dfrac{dH}{dt}\leqslant 0
\end{displaymath}
for all time. Moreover $\dfrac{dH}{dt} =0$, if and only if $f=\Meq$.
\end{theorem}

\begin{proof}

If $Z_r=1$, $\M=\Meq$, so the model reduces to a BGK model with a single temperature, for which the $\Htheo$ theorem holds. Let us suppose then that $Z_r>1$.

Since density, momentum and total energy are conserved, the equation for the Maxwellian is given by \eqref{eq:BGK_polyMax}.
Let us multiply equation \eqref{eq:BGK_poly} by $(1+\ln f)$ and integrate over the space  of all microscopic velocities  $\R^D$. Then multiply \eqref{eq:BGK_polyMax} by $Z_r(1+\ln\M)$, integrate in velocity space and add the two results. Using the fact that $\int f=\int\M= \int\Meq$, and adding and subtracting $(\M-f)\ln\M$, we obtain
\begin{equation} \label{eq:dHdt}
\dfrac{d H}{dt}=\dfrac{1}{\tau} \underbrace{ \int_{\R^D}
-\left(f-\M\right)\left(\ln{f}-\ln{\M}\right) \, \D v}_{A}+
\dfrac{1}{\tau} \underbrace{\int_{\R^D}
    \left(\Meq-f\right)\ln{\M} \, \D v}_{B}
\end{equation}
The term $A$ is clearly negative, due to the convexity of the $\log$ function. We continue evaluating the sign of $B$. To this end, we subtract the quantity $\int (\Meq-f) \ln{\Meq}$, which is zero, due to conservation of mass, momentum and total energy, so $B$ can be rewritten as
\begin{displaymath}
B = - \int_{\R^D} \big(\ln \M - \ln \Meq \big)(\Meq-f).
\end{displaymath}
Proving that $B\leq 0$ is a little technical, see \cite{Bernard2018Poly}.
Since both $A$ and $B$ are non-positive, their sum is also non-positive. Moreover, $A$ is zero if and only if $f=\M$, while $B$ is zero if and only if $\M=\Meq$ or $f=\Meq$. This implies that their sum is zero if and only if $f=\M=\Meq$, which means that at equilibrium $f$ is a Maxwellian with all temperatures equal to $\Teq$.
\end{proof}

\subsubsection{Positivity of the temperatures}

Consider again the space homogeneous case. Suppose that at the initial time $f(v,t=0)\geq0$. Then $T_t$ and $T_r$ at $t=0$ are both positive. We set up a ``well prepared'' initial condition, namely we set $\Lambda_t(t=0)=T_t(t=0)$ and $\Lambda_r(t=0)=T_r(t=0)$. Then all temperatures involved in the model remain positive for all time.

In fact, $\Teq=dT_t+rT_r$ remains constant at all times, and since it is positive at $t=0$, it will remain a positive number. Integrating the equation for $\Lambda_r$ \eqref{eq:Tr_closure}, with $u\cdot\nabla_x\Lambda_r\equiv 0$, one finds
\[
\Lambda_r(t) = \Lambda_r(0)e^{-\tfrac{t}{Z_r\tau}}+\Teq \big( 1- e^{-\tfrac{t}{Z_r\tau}}\big),
\]
which is a convex combination of $\Lambda_r(0)$ and $\Teq$, proving that $\Lambda_r$ remains positive for all time. Next, we multiply \eqref{eq:BGK_poly} by $\prod_{k=d+1}^{d+r}v_k^2$ and we integrate in velocity space. Using mass conservation one obtains the evolution of the rotational temperature
\[
\frac{\D T_r}{\D t} = \frac{1}{\tau}\big( \Lambda_r-T_r\big).
\]
Substituting the analytic solution for $\Lambda_r$ just computed, we obtain a linear non homogeneous ODE, which has solution
\[
T_r(t) = T_r(0) e^{-\tfrac{t}{\tau}} + \tfrac{Z_r}{Z_r-1}(\Lambda_r(0)-\Teq )
\big( e^{-\tfrac{t}{Z_r\tau}}-e^{-\tfrac{t}{\tau}} \big) + \Teq \big(1-e^{-\tfrac{t}{\tau}}\big).
\]
Substituting the well prepared initial condition, we find
\[
T_r(t) = T_r(0) \left( \tfrac{Z_r}{Z_r-1}e^{-\tfrac{t}{Z_r\tau}} - \tfrac{1}{Z_r-1}e^{-\tfrac{t}{\tau}} \right) + \Teq \left(1- \tfrac{Z_r}{Z_r-1}e^{-\tfrac{t}{Z_r\tau}} + \tfrac{1}{Z_r-1}e^{-\tfrac{t}{\tau}} \right) .
\]
Let $c(t;Z_r)$ be the coefficient in the first parenthesis. Clearly $c(0;Z_r)=1$, while $c(t;Z_r) \to 0$ as $t \to \infty$. Moreover $\partial_t c(t;Z_r)\leq 0$, provided $Z_r > 1$. Thus, $T_r(t)$ is a convex combination of positive numbers, and therefore remains positive for all times. The same argument applies also to $T_t$ and $\Lambda_t$.

\subsection{Chu's reduction}\label{Chu}

The polyatomic model we have introduced requires a large number of independent variables, because each new degree of freedom is associated to a new component for the vector $v$ of microscopic speeds. Thus, the computational complexity of the polyatomic model \eqref{eq:BGK_poly} increases dramatically with respect to the standard BGK model for a monoatomic gas, and becomes prohibitive if the number of internal modes $r$ is large. This complexity however can be drastically reduced using an approach proposed by Chu \cite{Chu1965}, which we will adapt in the following to the case of the polyatomic BGK model.

In  the standard BGK model, Chu's reduction can be applied whenever the distribution function $f$ depends only on $m<d$ degrees of freedom in space. Then it is possible to rewrite the kinetic equation using only $m$ degrees of freedom, also in the microscopic velocity space.
For example, in a two dimensional problem in space, the number of independent variables can be reduced to four plus time (two in space and two in microscopic velocity). 

We review Chu's reduction, outlining how it can be applied to the polyatomic model, reducing the computational complexity to $d$ independent variables in velocity space, instead of $d+r$, at the price of introducing one distribution function for each internal temperature. In the case of a bi-atomic molecule, we will apply the reduction to aggregate the internal energy degrees of freedom.
Let us consider the case in which we want to reduce the $r$  rotational degrees of freedom (\dof), while the system has $d$ traslational \dof, again with $d+r=D$. Let us label the  indices pertaining to the translational and the rotational \dof\ as $k=1,\dots,d$ and $k=d+1, \dots, d+r=D$ for the traslational and the rotational \dof, respectively. Correspondingly, the microscopic velocities will be partioned as $(\xi_1,\dots,\xi_{d},\eta_{1},\dots,\eta_{r})= (\bm{\xi},\bm{\eta})$, with $\xi_k=v_k, k=1,\dots,d$, and $\eta_k=v_{d+k}, k=1,\dots, r$. We introduce the two reduced distribution functions
\begin{align*}
&f_1(x,\bm{\xi},t)=\int_{\R^{r}} f(x,\bm{\xi},\bm{\eta},t) \;\D \bm{\eta}.  \\
&f_2(x,\bm{\xi},t)=\int_{\R^{r}}\sum_{k=1}^{r}\eta_k^2 f(x,\bm{\xi},\bm{\eta},t)\;\D \bm{\eta}.
\end{align*}
The model reduces to a system of two equations:
\begin{align}\label{eq:Chuf}
&\partial t f_1+\bm{\xi}\cdot\nabla_x f_1= \dfrac{1}{\tau}\Big(M_{f_1}-f_1\Big)\\
&\partial t f_2+\bm{\xi}\cdot\nabla_x f_2= \dfrac{1}{\tau}\Big(M_{f_2}-f_2\Big) \nonumber
\end{align}
where the reduced Maxwellians are expressed as:
\begin{align*}
&M_{f_1}(x,\bm{\xi},t)=\int_{\R^{r}}
\M(x,\bm{\xi},\bm{\eta},t)
\;\D \bm{\eta}\\
&M_{f_2}(x,\bm{\xi},t)=\int_{\R^{r}}\sum_{k=1}^{r}\eta_k^2 \;
\M(x,\bm{\xi},\bm{\eta},t)
\;\D \bm{\eta}
\end{align*}
Computing the integrals in the reduced velocity space $\R^r$, we find
\begin{align}\label{eq:ChuM}
&M_{f_1}(x,\bm{\xi},t)= n(x,t) \Big( \frac{m}{2\pi \Lambda_t}\Big)^{d/2}
\mbox{exp}\Big(-\frac{m}{2\Lambda_t}(\bm{\xi}-u)^{2}\Big)\\
&M_{f_2}(x,\bm{\xi},t)=\tfrac{r\Lambda_r}{m}\,M_{f_1}(x,\bm{\xi},t). \nonumber
\end{align}
where we have evaluated the integral which defines $M_{f_2}$, recalling that the expected value of the velocity on the reduced dimensions is zero. Note that $f_2$ and $M_{f_2}$ have the dimensions of a distribution function, times a velocity squared.

The macroscopic quantities needed to compute $M_{f_1}$ and $M_{f_2}$ are found from the conservation equations. One finds
\begin{align}\label{e:ChuMoments}
n(x,t) &= \int_{\R^d} f_1(x,\bm{\xi},t) \D \bm{\xi}  \\
u(x,t) &= \int_{\R^d} \bm{\xi}f_1(x,\bm{\xi},t) \D \bm{\xi} \nonumber \\
\frac{n(x,t)}{m} \left( d\Lambda_t(x,t) + r\Lambda_r(x,t)\right) &= \int_{\R^d} (\bm{\xi}-u)^2 f_1(x,\bm{\xi},t) \D \bm{\xi} + \int_{\R^d}f_2(x,\bm{\xi},t) \D \bm{\xi}. \nonumber
\end{align}
Note that $f_2$ is needed only to compute the temperature. Further, if $\Lambda_t=\Lambda_r$, the problem would be closed, using the two evolution equations \eqref{eq:Chuf}, with the two Maxwellians defined by \eqref{eq:ChuM}, through the moments computed in \eqref{e:ChuMoments}. If instead the two temperatures are different one needs also the relaxation equation for the temperatures \eqref{eq:Tr_closure} and the closure relation \eqref{e:EqTemperature}.

\subsection{A polyatomic BGK model with a single microscopic variable for the internal modes}

The analysis appearing in \cite{Andries:2001vi} instead is based on a {\em single} microscopic variable $\zeta$ which accounts for all the internal energy of the polyatomic gas. Thus, $f=f(x,v,\zeta,t)$, where $\zeta\in \R^+$, such that the internal energy per unit mass due to the rotational and vibrational modes is $\epsilon=\zeta^{2/r}$, where $r$ is the number of internal degrees of freedom, which is related to the  constant of polytropic gases $\gamma$, which is the ratio of the specific heats at constant pressure and constant volume, as $\gamma= \tfrac{r+5}{r+3}$. 

The macroscopic density, momentum and energy of the gas are given by
\begin{align}\label{e:macroPolyPerth}
\rho(x,t) & = m \int_{\R^d} \int_{\R^+} f(x,v,\zeta,t) \, \D v\, \D \zeta. \\
\rho u (x,t) & = m \int_{\R^d} \int_{\R^+} v f(x,v,\zeta,t) \, \D v\, \D \zeta. \nonumber \\
\tfrac12 \rho u^2 (x,t) + \rho e (x,t) & = m \int_{\R^d} \int_{\R^+}  \big( \tfrac12 v^2 + \zeta^{2/r}\big) f(x,v,\zeta,t) \, \D v\, \D \zeta. \nonumber \\
\end{align}

The internal energy is considered as the sum of two contributions,
\[
\rho e_t= m \int_{\R^d} \int_{\R^+}  \tfrac12 (v-u)^2 f(x,v,\zeta,t) \, \D v\, \D \zeta, \qquad \rho e_r= m \int_{\R^d} \int_{\R^+} \zeta^{2/r} f(x,v,\zeta,t) \, \D v\, \D \zeta,
\]
which give a traslational and a rotational temperature, namely
\[
T_t=\tfrac{2m}{d}e_t, \qquad T_r=\tfrac{2m}{r}e_r.
\]
Further, the model contains an {\em equilibrium} temperature $\Teq$ which is defined through the total internal energy,
\[
\Teq= \tfrac{2m}{d+r}e, \qquad \Longrightarrow \quad (d+r) \Teq = d T_t+rT_r.
\]
Further, a mechanism is needed to relax the partial temperatures on a global equilibrium temperature. This is realized through a relaxation temperature $T_{\mbox{\scriptsize rel}}$ which is defined through the convex combination 
$T_{\mbox{\scriptsize rel}}=\theta \Teq + (1-\theta) T_r$, with $0\leq \theta\leq 1$. Then, if the model in \cite{Andries:2001vi} is restricted for simplicity to the BGK case, the Maxwellian in the BGK operator is given by
\[
\M = \frac{n}{\lambda_r} \Big( \frac{m}{2\pi T}\Big)^{d/2}\Big( \frac{m}{ T_{\mbox{\scriptsize rel}}}\Big)^{r/2}\mbox{exp}\Big(  -\frac{m(v-u)^2}{2T}\Big)
\mbox{exp}\Big(  -\frac{\zeta^{2/r}}{T_{\mbox{\scriptsize rel}}}\Big),
\]
where $\lambda_r$ is the normalization constant $\lambda_r=\int \mbox{exp}(-\zeta^2/r) \, \D \zeta$. The relaxation of $T_t$ and $T_r$ towards $\Teq$ is given by a dependence on $\theta$ introduced in the relaxation time, and the BGK equation is modified to
\[
\transport{f} =\frac{1}{\tau(x,t,\theta)}\Big( \M - f\Big).
\]
In \cite{Andries:2001vi}, an $\Htheo$ theorem is proved for this model, and also the possibility of using a modified Chu's reduction to get rid of the internal energy microscopic variable. 

The two models just described have several similarities, but one can argue that the model in \cite{Bernard2018Poly} is more intuitive, because of the physical meaning of the internal degrees of freedom, and because the different relaxations are more readable.

\section{BGK for mixtures}

We now turn to the last class of models I will discuss in these notes. We will consider BGK models for mixtures of different gases. Here the difficulty is to represent the momentum and energy exchanges between the different species as the whole mixture evolves toward equilibrium.

There is a huge literature on this topic. I will restrict the discussion to non-reactive mixtures, for which the numbers of molecules for each component is constant in time. Models that account for chemical reactions can be found for instance in \cite{GroppiSpiga2013reacting}, \cite{Bisi2010} or \cite{Schneider20151075} and references therein. Models for mixtures often are designed with particular ends in mind. For instance \cite{Brull2012derivation} the accent is on the matching of transport coefficients.
See also \cite{Brull2014ESmixtures, BarangerBisiBrull2018PolyMixtures, Bisi2016PolyMixtures} for mixture models based on the ES-BGK paradigm, \cite{KlingPirnerPuppo2018} to include polyatomic effects. More theoretical aspects, as existence and uniqueness of solutions can be found in \cite{KlingPirner2018} or \cite{Pirner2018}.
The discussion that follows is based on \cite{KlingPirner:2017}, \cite{BobylevGroppi2018}, and \cite{AndriesAokiPerthame2002}.

The Boltzmann equation \eqref{e:Boltz} can be written as
\[
\transport{f} = Q^B(f,f),
\]
where $Q(f,f)$ denotes the collision term, and the notation underlines the fact that we are considering binary collisions.

If we have two different species, each described by a different distribution function $f_i, i=1,2$, then we can have only two types of interactions: interactions of a species with itself, giving rise to a collision term $Q_{ii}(f_i,f_i)$ and collisions with the other species, $Q_{ij}(f_i,f_j)$. Thus, for two species
\begin{align}\label{e:Boltz2}
\transport{f_1} & = Q_{11}(f_1,f_1) + Q_{12}(f_1,f_2) \\
\transport{f_2} & = Q_{22}(f_2,f_2) + Q_{21}(f_2,f_1). \nonumber
\end{align}
In other words, for a system with $k$ different species, we obtain a system of $k$ equations, with $k$ collision terms,
\begin{equation}\label{e:Boltzm}
\transport{f_i} = \sum_{j=1}^k Q{ij}(f_i,f_j), \qquad i=1,\dots,k.
\end{equation}
Therefore, the natural extension of this setting to the BGK operator is
\begin{equation}
\transport{f_i} = \sum_{j=1}^k \nu_{i,j} n_j({\M}_{i,j} -f_i), \qquad i=1,\dots,k.
\end{equation}

\subsection{BGK mixture models mimicking multispecies Boltzmann}
For simplicity, we will consider only two species. We obtain the model in \cite{KlingPirner:2017}:
\begin{align}\label{e:BGK2}
\transport{f_1} & = \nu_{1,1} n_1({\M}_{1} -f_1) + \nu_{1,2} n_2({\M}_{1,2} -f_1)\\
\transport{f_2} & = \nu_{2,2} n_2({\M}_{2} -f_2) + \nu_{2,1} n_1({\M}_{2,1} -f_2). \nonumber
\end{align}
The first term in both equations accounts for the interactions of each species with itself. Since, if the two species did not interact, i.e. $\nu_{1,2}=\nu_{2,1}=0$, we would still have off equilibrium effects, the inner species Maxwellians ${\M}_i$ must be the same we would have for the case of a single species, or
\begin{equation}\label{e:innerMaxw}
{\M}_i(x,v,t) =  \frac{n_i(x,t)}{(2\pi T_i(x,t)/m_i)^{d/2}} {\Large e^{-\frac{m_i(v-u_i(x,t))^2}{2T_i(x,t)}}}, \qquad i=1,2.
\end{equation}
The remaining Maxwellians contain the effects of the interspecies interactions. They drive the system to a state in which  the single species velocity and temperature relax to a mixture velocity and temperature: 
\begin{equation}\label{e:interMaxw}
{\M}_{i,j}(x,v,t) =  \frac{n_i(x,t)}{(2\pi T_{i,j}(x,t)/m_i)^{d/2}} e^{-\frac{m_i(v-u_{i,j}(x,t))^2}{2T_{i,j}(x,t)}}, \qquad i=1,2, j\neq i.
\end{equation}
First, we turn to the collision frequencies. 
We observe that close to equilibrium all temperatures will be of the same order of magnitude, while the thermal speeds of the particles will be ${u_T}_i \approx \sqrt{2T_i/m_i}$. The ratio between the thermal speeds of the two species is then given by ${u_T}_i/{u_T}_j \approx \sqrt{m_j/m_i}$. The two species may be characterized by very different masses, as in the case of plasmas, where if one takes the species one to coincide with the protons  and species 2 with electrons $m_2/m_1\simeq 1/2000<<1$. In this case one typically has
\begin{equation}\label{e:collfreq}
\nu_{21}\approx \nu_{22}\approx \sqrt{\frac{m_1}{m_2}}\nu_{11}\approx \frac{m_1}{m_2} \nu_{12},
\end{equation}
so that we can take $\nu_{1,2}/\nu_{2,1}=m_2/m_1$. For more details, see \cite{KlingPirner:2017}. In general, one can assume that the collision frequencies $\nu_{ij}$ can be measured or estimated, to reduce the number of free parameters.

The construction of the Maxwellians ensures that mass is conserved. In fact, computing the integral in velocity space of equations \eqref{e:BGK2}, with the Maxwellians defined in \eqref{e:innerMaxw} and \eqref{e:interMaxw}, one easily obtains
\[
\partial_t \rho_i + \nabla (\rho_i u_i) = 0, \qquad i=1,2.
\]
In the case of chemical reactions, when the mass of each species is not conserved, the definitions of the interspecies Maxwellians must be modified to permit an exchange of mass between the two species, see for instance \cite{Bisi2010}. In all cases, the decay towards equilibrium of the mixture is based on an exchange of momentum and energy between the two species.

Now, to continue, one needs assumptions on the structure of the macroscopic quantities appearing in the definition of the Maxwellians. 
\begin{assumption}\label{Ass:speed}
The mixture velocity $u_{12}$ is a convex combination of the single species macroscopic speeds:
\begin{equation}\label{e:intraspeed}
u_{12} = \delta u_1+ (1-\delta) u_2, \qquad 0\leq \delta \leq 1.
\end{equation}
\end{assumption}
The assumption on the mixture temperature is slightly more complex because energy exchange is due not only to a heat flux between the different species, caused by their temperature differences, but also to a transfer of kinetic energy at the macroscopic level.
\begin{assumption}\label{Ass:temp}
The mixture temperature $T_{12}$ is a convex combination of the single species macroscopic temperatures, plus a term which accounts for the exchange of macroscopic kinetic energy:
\begin{equation}\label{e:intratemp}
T_{12} = \alpha T_1+ (1-\alpha) T_2 + \gamma|u_1-u_2|^2, \qquad 0\leq \alpha \leq 1, \gamma>0.
\end{equation}
Note that $\alpha$ is a pure number, while $\gamma$ has the dimensions of mass.
\end{assumption}
Typically, all parameters comparing in the assumptions will depend on the momentum and energy transfer occurring at the microscopic level, when molecules from different species collide, and therefore one expects that they will be linked to the mass ratio between the two species.

Once $u_{12}$ and $T_{12}$ have been set using Assumptions \ref{Ass:speed} and \ref{Ass:temp}, the remaining quantities $u_{21}$ and $T_{21}$ are fixed imposing momentum and energy conservation for the whole mixture. For momentum
\[
m_1\int_{\R^d} \nu_{1,2}n_2\left( {\M}_{1,2}-f_1\right)v + m_2\int_{\R^d} \nu_{2,1}n_1\left( {\M}_{2,1}-f_2\right)v = 0,
\]
while for the energy
\[
m_1\int_{\R^d} \nu_{1,2}n_2\left( {\M}_{1,2}-f_1\right)v^2 + m_2\int_{\R^d} \nu_{2,1}n_1\left( {\M}_{2,1}-f_2\right)v^2 = 0,
\]
where I have used the fact that ${\M}_i$ and $f_i$ share the same velocity and temperature, so the inner species interactions - correctly - do not produce a macroscopic flux of momentum and energy.
One easily obtains
\begin{equation}\label{e:intraspeed2}
u_{2,1} = \left( 1- \varepsilon(1-\delta)\right)u_2 + \varepsilon(1-\delta)u_1, \qquad \varepsilon = \frac{m_1\nu_{1,2}}{m_2\nu_{2,1}}.
\end{equation}
Note that $\varepsilon>0$, and if $\varepsilon<1$, this is still a convex combination of $u_1$ and $u_2$. On the other hand, if $\varepsilon>1$, it is enough to exchange the roles of $f_1$ and $f_2$. The formulas for the two interspecies velocities become completely symmetric if  $\varepsilon=1$. For the interspecies temperature, the expression is more complicated, and one finds
\begin{equation}\label{e:intraTemp2}
\begin{split}
T_{2,1}  &= \left(1-\frac{\nu_{1,2}}{\nu_{2,1}}(1-\alpha) \right) T_2 + \frac{\nu_{1,2}}{\nu_{2,1}}(1-\alpha) T_1 -\gamma \frac{\nu_{1,2}}{\nu_{2,1}}(u_1-u_2)^2 \\
& + \frac{\varepsilon}{d\nu_{2,1}}(1-\delta) \left[ m_2\nu_{2,1}- m_1\nu_{1,2}+ \delta(m_2\nu_{2,1}+m_1\nu_{1,2}) \right](u_1-u_2)^2.
\end{split}
\end{equation}
If we assume $m_2>m_1$ and the relations between the collision frequencies \eqref{e:collfreq} hold, then the first part of the equation is again a convex combination of $T_1$ and $T_2$. Further, the expression in the square parenthesis is positive, so $T_{2,1}>0$ if $\gamma$ is small enough. 

\begin{remark}
If $u_1=u_2=U$, then clearly $u_{1,2}=u_{2,1}=U$ and there is neither momentum nor kinetic energy exchange between the species. If further $T_1=T_2=T$, then $T_{1,2}=T_{2,1}=T$ and there are no heat exchanges.
\end{remark}

For the space homogeneous case, in \cite{KlingPirner:2017} we prove that $f_i(v,t) \geq 0$ for all times, provided the initial data $f_i(v,t=0)\geq 0, i=1,2$. For the general case, see \cite{KlingPirner2018}. Further, if $f_1$ and $f_2$ are positive, all temperatures are also positive, under mild conditions on the coefficient $\gamma$.

The approach described above, developed in \cite{KlingPirner:2017} {\em assumes} that the interspecies velocities and temperatures can be written as \eqref{e:intraspeed} and \eqref{e:intratemp}. Further, the model is complete when the 4 collision frequencies $\nu_{i,j}, i,j=1,2$ are known, together with the parameters $\delta, \alpha$ and $\gamma$. This approach has been further studied in \cite{BobylevGroppi2018}. The idea there is to {\em compute} the parameters $\delta, \alpha$ and $\gamma$, imposing that momentum and energy transfers between the species  in the BGK setting reproduce the macroscopic momentum and energy exchanges derived from the full Boltzmann equation. Thus
\[
\begin{split}
\nu_{i,j} n_j & \int_{\R^d}[{\M}_{i,j}(v) -f_i(v)] \;\phi(v)\;\D v \\ & = \int_{\R^d} \int_{\R^d}\int_{S^+}B_{i,j}(|v-v_*|,\hat{n})[ f_i(v')f_j(v_*') -  f_i(v)f_j(v_*)] \; \phi(v)\;\D v_*\; \D v\;\D\omega,
\end{split}
\]
where $\phi(v)$ denotes one of the collision invariants, i.e. $\phi(v)= 1, v, v^2$.
Note that this request results automatically in mass, momentum and energy conservation, because these conservation principles hold at the Boltzmann level. Furthermore, the interspecies velocities and temperatures are now {\em given} quantities, provided one is able to compute the right hand side. More precisely, let $Q^B_{i,j}$ be the Boltzmann collision term between the species $i$ and $j$, then
\[
\nu_{i,j}n_jn_i(u_{i,j}-u_i) = \int_{\R^d} Q_{i,j}^B v\; \D v
\]
and 
\[
\nu_{i,j}n_jn_i\left(d\frac{T_{i,j}-T_i}{m_i}+ ||u_{i,j}||^2-||u_i||^2\right) = \int_{\R^d} Q_{i,j}^B v^2\; \D v,
\]
which define $u_{i,j}$ and $T_{i,j}$ in terms of integrals of the Boltzmann collision term.

This task can be carried out exactly in the case of Maxwellian molecules, when, as in BGK, the collision cross section does not depend on the relative speed between the particles. In this case, the integrals on the right hand side can be explicitly computed, giving
\[
\nu_{i,j}n_jn_i(u_{i,j}-u_i) = -\frac{m_j}{m_i+m_j}\lambda_{i,j}n_in_j (u_i-u_j)
\]
for the velocity, while for the temperature
\[
\begin{split}
\nu_{i,j}n_jn_i  &\left(d\frac{T_{i,j}-T_i}{m_i}+ ||u_{i,j}||^2-||u_i||^2\right)\\ & = \frac{-2m_j}{(m_i+m_j)^2}\lambda_{i,j}n_in_j [d(T_i-T_j)+ (m_iu_i+m_ju_j)(u_i-u_j)].
\end{split}
\]
The parameters $\lambda_{i,j}$ characterize the interaction, and they are given by
\[
\lambda_{i,j} = \int_{S^+}B_{i,j}(|v-v_*|,\hat{n}) \D\omega = \lambda_{j,i}.
\]
In this way, it is confirmed that the interspecies velocities and temperatures can be written as in \eqref{e:intraspeed} and \eqref{e:intratemp}, at least in the case of Maxwellian molecules, where the coefficients $\delta, \alpha$ and $\gamma$ can be obtained from the expression of the cross sections $\lambda_{i,j}$.

For more general cases, the exchange terms in BGK cannot represent exactly the exchange terms obtained with the Boltzmann model. Here Bobylev et al. in  
\cite{BobylevGroppi2018} suggest to average the effect of the velocity-dependent cross section, choosing to reproduce exactly one global feature. 

The important advancement obtained in \cite{BobylevGroppi2018} is that this work gives tools to compute the coefficients in the interspecies exchange terms, since both $u_{i,j}$ and $T_{i,j}$ are completely defined in terms of $u_i, u_j, T_i, T_j$ and the global cross sections $\lambda_{i,j}$.

\begin{remark}{\bf Multispecies mixtures}

\noindent It should be noted that the models described above for a mixture composed of two gases can be extended easily to a multicomponent mixture with $n$ different species. In fact, the structure of the Boltzmann integral implies that
\[
\int_{\R^d} Q_{i,j}^B \, \phi(v) + \int_{\R^d} Q_{j,i}^B \, \phi(v) =0,
\]
for $\phi(v)=1, v, ||v||^2$.
\end{remark}

\subsection{$\Htheo$ theorem for mixtures}

The $\Htheo$ theorem was already proven in \ref{H-theorem} for the standard BGK model. Here we extend the proof to the case of binary mixtures. The $\Htheo$ theorem states that the entropy {\em decreases} with time, and reaches a minimum for $f=\M$. Since in the space homogeneous case, $n$, $u$ and $T$ are constant, then \eqref{e:Htheo} states that the entropy decreases until the system relaxes on the Maxwellian $\M$ which has the same moments of the initial distribution $f(x,v,t=0)$. 

We can prove a similar results for mixtures. Again, we restrict ourselves to the space homogeneous case. For the general case, see \cite{KlingPirner2018}.

\begin{theorem}[$\Htheo$ theorem for mixtures]
Consider the system \ref{e:BGK2} in the space homogeneous case, with $u_{12}$ and $T_{12}$ defined by \eqref{e:intraspeed} and \eqref{e:intratemp}, while  $u_{21}$ and $T_{21}$ are fixed by conservation of total momentum and energy as in \eqref{e:intraspeed2} and \eqref{e:intraTemp2}. Then the total entropy of the system
\begin{equation}
H = \int_{\R^d} f_1 \log f_1 \, \D v + \int_{\R^d} f_2 \log f_2 \, \D v 
\end{equation}
decreases with time.

The minimum of $H$ is reached when both $f_1$ and $f_2$ are Maxwell distributions with the same temperature and macroscopic velocity.
\end{theorem}
\begin{proof}
We need to evaluate
\begin{align*}
\partial_t H & = \nu_{11}\, n_1\int_{\R^d} \log f_1({\M}_{1} -f_1)\, \D v
             + \nu_{12}\, n_2\int_{\R^d} \log f_1({\M}_{1,2} -f_1)\, \D v \\
             & + \nu_{22}\, n_2\int_{\R^d} \log f_2({\M}_{2} -f_2)\, \D v
             + \nu_{21}\, n_1\int_{\R^d} \log f_2({\M}_{2,1} -f_2)\, \D v,
\end{align*}
where mass conservation has been used.
The proof relies on the fact that for the single species inequality \eqref{e:Htheo} holds, thus you only need to study the sign of the expression
\[
S(f_1,f_2) = \nu_{12}\, n_2\int_{\R^d} \log f_1({\M}_{1,2} -f_1)\, \D v +
     \nu_{21}\, n_1\int_{\R^d} \log f_2({\M}_{2,1} -f_2)\, \D v,
\]
because the remaining terms are already known to be negative. To prove that $S(f_1,f_2)\leq 0$, the inequality
\[
\nu_{12} \log T_{1,2} + \nu_{21} \log T_{2,1} \geq \nu_{12} \log T_{1} + \nu_{21} \log T_{2}
\]
is needed. This is a bit tricky, and uses the expressions for $T_{1,2}$ and $T_{2,1}$, and of the interspecies velocities. All details can be found in \cite{KlingPirner:2017}.

We just underline that $H$ remains constant in time if and only if the single species terms cancel, which implies that $f_1={\M}_1$ and $f_2={\M}_2$, and the interspecies entropy production $S(f_1,f_2)$ is also zero, which requires $f_1={\M}_{1,2}$ and $f_2={\M}_{2,1}$. This in turn implies $u_{1,2}=u_1$, $u_{2,1}=u_2$, $T_{1,2}=T_1$ and $T_{2,1}=T_2$. Finally, applying the equations for the mixture velocities and temperatures to this case, one easily finds $u_2=u_1$ and $T_2=T_1$, which completes the proof.
\end{proof}

The $\Htheo$-theorem for mixtures gives details on the global equilibrium of the system. It states that the system is in equilibrium if and only if both distributions are Maxwellians, with the same macroscopic speed, and with the same temperature. This means also that at equilibrium no net macroscopic momentum or heat exchange between the two components of the mixture can take place.

\subsection{Single collision term models}

In \cite{AndriesAokiPerthame2002}, Andries, Aoki and Perthame published a paper in which the BGK model for a mixture of gases is written with a 
single collision term (AAP model, in the following). This idea has been followed by several researchers, such as Bisi, Brull,  Groppi, Spiga, and others, see for instance \cite{BarangerBisiBrull2018PolyMixtures},  \cite{Bisi2010}, \cite{Bisi2016PolyMixtures} and their references. The AAP model is written as
\[
\partial_t f_i+ v \cdot \nabla f_i = Q_i:= \frac{1}{\tau_i}(M_i-f_i),\quad i=1,\dots,N,
\]
where $N$ is the number of species, $M_i$ are the $N$ Maxwellians depending each on a {\em mixture velocity $\tilde{u}_i$} and a {\em mixture temperature $\tilde{T}_i$} defined in such a way to match the momentum and energy transfer in Boltzmann equation, for the hard sphere model. This choice is due to the fact that in the hard sphere model the Boltzmann cross-section does not depend on the particles velocities, and, for this reason, it is closer to the BGK setting.

Writing Boltzmann for the hard sphere model, the AAP model is based on the following equations for the unknowns $\tilde{u}_i$ and $\tilde{T}_i$:
\begin{align*}
\int m_i v Q_i& =  2\sum_j\mu_{ij}\chi_{ij}n_in_j (u_j-u_i) = \frac{1}{\tau_i}m_i n_i (\tilde{u}_i-u_i) \\
\int \tfrac12 m_i v^2 Q_i& = 
 4 \sum_j \chi_{ij}n_j \frac{\mu_{ij}}{m_i+m_j}\left[\tfrac32(T_j-T_i)+ \tfrac{m_j}{2}(u_j-u_i)^2\right] \\
& \qquad = \frac{1}{\tau_i}n_i\left[\tfrac12 m_i(\tilde{u}_i^2-u_i^2) + \tfrac32 (\tilde{T}_i-T_i)\right]  
\end{align*}
where 
\[
\mu_{ij}= \frac{m_i m_j}{m_i+m_j}
\]
is the reduced mass for the species $i$ and $j$, and the terms $\chi_{ij}=\int_{S^+}B_{ij}(\hat{n})$ are the mixed collision cross sections. 
The relations above define the new quantities $\tilde{u}_i$ and $\tilde{T}_i$, for each species.

The AAP single collision term mixture model can also be proven to satisfy an H-theorem. Moreover, the AAP kinetic model for mixtures was constructed in order to satisfy the {\em indifferentiability principle}. For the sake of simplicity, consider two species, such that  $m_i=m_j=m$ and $\tau_i=\tau_j=\tau$, which means that the two species are dynamically identical. A model is said to satisfy the indifferentiability principle if, in the case of identical particles,  the  model can be reduced to a {\em single} species BGK model for the unknown
\[
f= f_1+f_2, \quad \partial_t f + v\cdot \nabla f=\frac{1}{\tau}(M-f).
\]
In fact one can prove that in this case
\[
\partial_t f_i + v\cdot \nabla f_i=\frac{1}{\tau}(n_i\tilde{M}-f_i)
\]
with $\tilde{M}= \sqrt{\big(\frac{m}{2\pi\tilde{T}}\big)^d}\exp{\left[-(v-\tilde{u})^2/(2\tilde{T}/m)\right]}$, because $\tilde{u_i}= \tilde{u_j}=\tilde{u}$ and $\tilde{T_i}= \tilde{T_j}=\tilde{T}$. Then $M$ is simply $M=\sum_in_i\tilde{M}=n\tilde{M}$.
Thus, the AAP model for identical particles reduces to a single species model, when all microscopic parameters coincide.

It is easy to see that for the mixture model I described before, the indifferentiability principle holds {\em only} at equilibrium. In  fact the evolution equations for identical particles are
\begin{align*}
 \partial_t{f_1} + v \cdot \nabla f_1 & = \frac{n_1+n_2}{\tau} \left(\frac{n_1}{n_1+n_2}{\M}_1+\frac{n_2}{n_1+n_2}{\M}_{1,2}-f_1\right)   \\
 \partial_t{f_2} + v \cdot \nabla f_2 & = \frac{n_1+n_2}{\tau} \left(\frac{n_2}{n_1+n_2}{\M}_2+\frac{n_1}{n_1+n_2}{\M}_{2,1}-f_2\right) 
 \end{align*}
and these reduce to a single species BGK model only when $u_1=u_2=u$, so that $u_{1,2}=u_{2,1}=u$ 
and $T_1=T_2=T$, which implies $T_{1,2}=T_{2,1}=T$, that is, only when the {\em global} equilibrium is achieved.
In other words the model in \cite{KlingPirner:2017}, when the particles are identical, reduces to a single species model only when the global equilibrium is reached.

\section{Asymptotic Preserving schemes for BGK models and mixtures}

In this section I'll discuss the main aspects of the numerical integration of \eqref{e:BGK}. I will start from the issues at the basis of the numerical integration of the standard BGK model. Then I will illustrate how the same techniques can be extended to the case of the other models described in these notes: ES-BGK, polyatomic, and mixtures. I will consider finite volume schemes, starting with a simple first order discretization, in one space dimension, and then sketch how the extension to high order schemes can be carried out. 
The one dimensional BGK equation \eqref{e:BGK} reduces to
\begin{equation}\label{e:1dBGK}
\partial_t f + v_1 \partial_x f = \nu n \, (\M-f),
\end{equation}
where with $v_1$ I denote the first component of the vector $v=[v_1,v_2,\dots v_d]$. Note that the integrals giving the moments of $f$ will still be computed on the whole $\R^d$, although it is possible to drop the dependence on $v_2\dots v_d$ using two distribution functions, each depending on $v_1$ and $x$ only, as described in \S \ref{Chu}. More details can be found in \cite{Chu1965}.

The first task is the construction of a grid in space and velocity. For simplicity I will consider a uniform grid, with mesh width $\Dx$ in space, $\Dv$ in velocity and $\Dt$ in time. Let $x_j=j\Dx, \, j\in \mathbb{Z}$, $v_k=k\Dv,\, k\in \mathbb{Z}^d$ and $t^n=n\Dt, \, n \in \mathbb{N}$ denote the grid points in space, velocity and time, for the one dimensional case. 

\subsection{Space and velocity discretization}

We start from the expression for a semidiscrete in time numerical scheme, in one space dimension. Here $f_{j,k}=f(x_j,v_k,t)$ denotes the function $f$ evaluated at the space velocity grid point $(x_j,v_k)$ at a fixed time $t$, $x_j$ in $\R$, $v_k\in \R^d$. Then
\[
\partial_t f_{j,k} = - (v_k)_1 \left. D_x f\right|_{j,k}(t) + \nu n\left( {\M}_{j,k}(t) - {f}_{j,k}(t) \right),
\]
where $D_x f|_{j,k}$ denotes the discrete space derivative at $(x_j,v_k)$, and $(v_k)_1$ is the first component of the vector $v_k$. Once the discretizations in space and velocity have been caries out, the BGK equation reduces to a system of ODE's in time. Since $v_k$ is constant with respect to $x$, each equation in the system of ODE's is a linear advection equation, with constant scalar speed $(v_k)_1$, plus a source term. We can approximate the space derivative with the upwind scheme, see \cite{LeVeque:book} for an introduction to the construction of numerical methods for equations of this type.
The upwind space discretization yields 
\begin{equation*}
\begin{split}
\partial_t f_{j,k}  & =  - \frac{1}{\Dx} \left( (v_k)_1^+ (f_{j,k}(t) - f_{j-1,k}(t)) + (v_k)_1^- (f_{j+1,k}(t) - f_{j,k}(t)) \right) \\
& + \nu n\left( {\M}_{j,k}(t) - {f}_{j,k}(t) \right),
\end{split}
\end{equation*}
where $v^+=\max(v,0)$ and $v^-=\min(v,0)$. Since $v^+=\tfrac12 (|v|+v)$ and $v^-=\tfrac12 (v-|v|)$, we can rewrite the expression above as
\begin{equation}\label{e:semidis}
\begin{split}
\partial_t f_{j,k}  & =  - (v_k)_1\frac{1}{2\Dx} \left( f_{j+1,k}(t) - f_{j-1,k}(t)\right) +  \frac{1}{2\Dx}|v_k|_1 \left(f_{j+1,k}(t) - 2f_{j,k}(t)+ f_{j-1,k}(t) \right) \\
& + \nu n\left( {\M}_{j,k}(t) - {f}_{j,k}(t) \right).
\end{split}
\end{equation}
The second term on the right hand side is a discretization of a term of the form $f_{xx}$. This  is a viscosity term, which is not present in the exact equation. It is an {\em artificial diffusion} term needed to obtain a stable scheme. Its effect is that steep profiles will be smoothed out. To reduce the unwanted artificial diffusion, while maintianing stability, a higher order accurate space discretization is needed. This aspect will be considered later. 
 Now there are two problems:
\begin{enumerate}
\item The computation of the Maxwellian requires the evaluation of the moments of $f$, which, in the discrete case, must be performed by quadrature. How should the quadrature be carried out, in order to guarantee that the moments of $f$ still satisfy the conservation laws \eqref{e:MassCons}, \eqref{e:MomCons} and \eqref{e:ECons}?
\item Once  the fully discrete system is obtained, one would like that the resulting numerical scheme, as the equilibrium is approached, should become a consistent numerical scheme for the equilibrium equations, which in this case are \eqref{e:Euler}. Beside, equilibrium is reached in a time scale $\sim \tau =(\nu n)^{-1}$. But when $\tau \to 0$,  the system becomes stiff. Is it possible to choose a time advancement method in such a way that the time step $\Delta t$ should not be chosen so small as to resolve the transient to equilibrium, but still guarantee that the exact equilibrium is obtained? This is idea will be formalized with the {\em asymptotic preserving} property. In other words, it is desirable that the correct equilibrium is found, even for $\Delta t >> \tau$.
\end{enumerate}

In applications, usually the space domain is limited, say $x\in \Omega \in \R^d$, where $\Omega$ is a bounded set with (hopefully) a smooth enough boundary. Then one needs to apply boundary conditions on the boundary of $\Omega$. I will not deal with boundary conditions here, but surely since the computation must be carried out on a bounded set, I will suppose that the mesh consists of a finite number $N$ of space cells. Consult \cite{Bernard2015BC} for a discussion on the technical aspects of boundary conditions for BGK type models, or \cite{FilbetJin2012bc} for treating high order accurate boundary conditions. 

Problems are of different kind for the velocity space, where all integrals needed to compute the moments of $f$ are defined on the whole of $\R^d$. For the integrals in velocity space, we need a quadrature rule.

Let 
\begin{equation}\label{e:quad}
<f> = \sum_k f(x,v_k,t) w_k \simeq  \int_{\R^d} f\, \D v
\end{equation}
be the quadrature rule chosen to approximate the integrals in velocity space, where $v_k$ are the nodes of the quadrature rule, and $w_k$ the corresponding weights. Suppose the total number of nodes is $K$. Then, we need to update $f$ only at the $K$ velocity nodes, in order to be able to evolve the moments, and therefore the Maxwellian. This means that the semidiscrete system \eqref{e:semidis} consists of $K N$ first order differential equations. In three dimensional computations, is should be clear that this number can be very large. In a typical application, one could have 50 nodes per direction in velocity space and 100 nodes per unit length in each space direction. This would give $50^3\times 100^3 =125 \times 10^9$ grid nodes on the unit cube. Therefore, it is important to reduce as much as possible the number of velocity nodes. 

One possibility is to use Gauss Hermite quadrature. This quadrature integrates functions on unbounded intervals of the form
\[
\int_{\R} s^\ell f(s) \, e^{-s^2} \, \D s \approx \sum_{k=1}^K s_k^\ell f (s_k) w_k,
\]
using well chosen nodes $s_k$ and weights $w_k$. This formula is exact for polynomials of degree $\ell$ with $2\ell -1=K_x$, where $K_x$ is the number of nodes in one direction of velocity space. This type of quadrature is particularly suited, apparently, for functions close to be Maxwellians, since to compute the moments of $f$ we need to be able to integrate exactly functions of the form above for $\ell=0, 1, 2$, and thus only three nodes in each direction in velocity space would be enough to ensure an exact evaluation of the moments. Moreover, for $K_x=3$, the velocity nodes are relatively small, which means that the system \eqref{e:semidis} is not stiff with respect to velocity space. The problem, however is that the weight function in Gaussian quadrature is the unknown Maxwellian, and it is centered in the unknown macroscopic velocity $u$, with standard deviation given by the unknown temperature $T$. Thus the location of the nodes is unknown, because they are distributed around the local value of the velocity $u(x,t)$, at a distance from $u$ that depends on the local value of $T(x,t)$, thus they change in space and time. For these reasons, Gauss-Hermite quadrature is not the most common strategy to compute the moments of $f$. See any text in numerical analysis, as \cite{StoerBulirsch} for Gaussian quadrature and its properties.

Another approach uses the trapezoid rule. In this case, it is necessary to fix bounds $-L,L$ such that all velocity nodes are contained in the box $[-L,L]^d$. The bounds may depend on time as in \cite{Bernard2014LocalGrids}, but usually it is convenient to keep a uniform grid spacing $\Dv$. In fact, the quadrature error for the trapezoid rule with a uniform grid for a function $f\in C^m(-L,L)$ in one direction is given by the Euler-McLaurin formula:
\begin{align}\label{eq:Euler_mclaurin}
\int_{-L}^L f(v)\; \D v & - \frac{\Dv}{2}\sum_k\left( f_k+f_{k+1}\right) 
            = 
         \sum_{l=1}^{m/2-1} \Dv^{2l} \frac{C_l}{(2l)!}\left( f^{(2l-1)}(L) - f^{(2l-1)}(-L)\right) \\
            & +  \Dv^{m} \frac{2C_{m}L}{(m)!} f^{(m)}(\xi), \qquad \xi \in (-L,L). \nonumber
\end{align}
where $C_m$ and all the $C_l$'s are suitable coefficients.
For the proof see \cite{StoerBulirsch}. 
The formula above shows that if the function $f$ is $2L$ periodic, together with its derivatives up to the order $m$, the boundary terms in the error cancel, and convergence becomes very fast. In our case, it is reasonable to suppose that $f$ decays very fast at infinity (remember: we are close to equilibrium, and $f$ is approaching a Maxwellian). So, if one takes $L$ large enough, $f$ and its derivatives will be very small at $\pm L$, and $f$ can be supposed to be $2L$ periodic, together with its derivatives.

However, once the quadrature rule in velocity space is chosen, the discrete moments of  $f$ will not be exact. In particular
\[
\ket{\phi(v) f} - 
\ket{\phi(v) \M} \, \neq \, 0, \qquad \phi(v) = [1, v, \tfrac12 v^2]^T.
\]

This means that at the discrete level, the right hand side of the conservation laws \eqref{e:MassCons}, \eqref{e:MomCons} and \eqref{e:ECons} will not be zero, and as $\tau \to 0$, with $\tau=1/\nu n$, one does not obtain the Euler equations \eqref{e:Euler}. This prompts the need for the computation of a discrete Maxwellian $\disM$ such that
\[
\ket{\phi(v) f} - 
\ket{ \phi(v) \disM} = 0,
\]
which are $d+2$ non linear equations.
The existence of the discrete Maxwellian $\disM$ and the way to compute it have been proposed in \cite{Mieussens2000Discrete}. The main result is that a $d+2$ dimensional vector $\mathbf{\alpha}$ exists, such that 
\begin{equation}\label{e:DiscreteMaxwellian}
\sum_k  w_k \phi(v_k) \left( f(x,v_k,t) -\exp( \alpha(x,t) \cdot[1, (v_1)_k,\dots,  (v_d)_k, \tfrac12 ||v_k||^2] \right)= 0,
\end{equation}
with $\phi(v)=[1, (v_1),\dots,  (v_d), \tfrac12 ||v||^2]$ as usual to denote the collision invariants. 
This defines, for any grid point in space, and any time, a system of $d+2$ non linear algebraic equations in the $d+2$ unknowns $\mathbf{\alpha}$, which has a unique solution, provided the speed lattice is large enough. We will call $\Mop$ the non linear operator which maps the discrete moments $U=\ket{\phi(v) f}$ into the discrete Maxwellian obtained solving the system \eqref{e:DiscreteMaxwellian}, $\Mop(U)=\disM$. For more details, see \cite{Mieussens2000Discrete}, while in \cite{Mieussens2000DiscreteCyl} one can find a fully implicit version of the resulting scheme. 
See also \cite{Gamba2009}.

In \cite{Mieussens2000Discrete} the proof of the existence of the discrete Maxwellian and its explicit construction are based on the characterization of the Maxwellian as the solution of the constrained minimization problem \eqref{e:minimo}, where the space of the constraints of the exact Maxwellian is substituted with the approximate constrains resulting from the discrete moments of $f$.

\subsection{Time discretization and AP schemes}

Finally, we must consider the time discretization. We rewrite \eqref{e:BGK} as
\begin{equation}\label{e:BGKtau}
\transport{f} = \frac{1}{\tau} (\M -f), \qquad \tau = \frac{1}{\nu n}
\end{equation}
to underline the fact that the system is stiff when the relaxation time $\tau \to 0$, and the system approaches equilibrium. Then, if one wants to deal with this regime, the numerical scheme must be implicit. Implicit schemes permit to use large time steps $\Dt>\tau$, and even $\Dt>>\tau$, without the need to resolve the fast transients which occur in a scale of order $\tau$. However, in the mean time one wants to capture the correct equilibrium. This is the purpose of AP (Asymptotic Preserving) schemes.

Formally, the idea is the following. Let us call $K^{\tau}$ the kinetic model we are considering, depending on the (small) parameter $\tau$. As $\tau \to 0$, the kinetic model approaches the equilibrium $M$, i.e., in our case, the macroscopic equations \eqref{e:Euler}. In symbols we write $K^{\tau} \to M$ as $\tau \to 0$. Suppose further that we have written a convergent scheme for the kinetic model, depending on some grid spacing $h$. We call $K^\tau_h$ the discretization of the model $K^{\tau}$. Since the scheme is convergent, we have $K^\tau_h \to K^\tau$, as $h \to 0$. Then, as $\tau \to 0$, we obtain in the limit the equilibrium model $M$. This is illustrated on the left of Fig. \ref{f:AP}. A scheme is AP if the discretization $K^\tau_h$ becomes, as $\tau \to 0$, a consistent discretization of the macroscopic equations, $M_h$. Thus it is not necessary to resolve the fast scale, i.e. computing the solution of $K^\tau_h$ for $h\to 0$ to have an approximation to the equilibrium $M$. The different limits involved are summarized in figure \ref{f:AP}.

\begin{figure}
\begin{center}
\begin{minipage}{0.45\textwidth}
 \vspace{0pt}
{\Large
\begin{tikzcd}
  K^{\tau}_h  \arrow[d, "h \to 0"]
    & 
    \\
  K^{\tau} \arrow[r,  "\tau \to 0" blue] &M
\end{tikzcd}
}
\end{minipage}
\begin{minipage}{0.45\textwidth}
 \vspace{0pt}\raggedright
{\Large
\begin{tikzcd}
  K^{\tau}_h \arrow[r, "\tau \to 0" red] \arrow[d, "h \to 0"]
    & M_h \arrow[d, "h \to 0" ] \\
  K^{\tau} \arrow[r,  "\tau \to 0" blue] &M
\end{tikzcd}
}
\end{minipage}
\end{center}
\caption{Illustration of AP schemes. Left: we discretize a kinetic model $K^{\tau}$, obtaining the scheme $K^\tau_h$. As $h \to 0$, the scheme $K^\tau_h$ produces a numerical solution that converges to the exact solution of $K^{\tau}$, which, for $\tau \to 0$, converges to a solution of the macroscopic equations $M$. Right: AP scheme. Here the discretization $K^{\tau}_h$ converges to a discretization of the macroscopic system $M_h$.}\label{f:AP}
\end{figure}
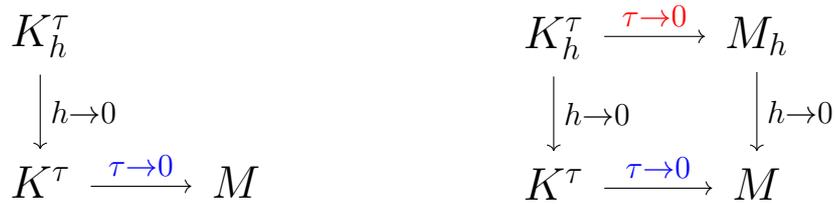

Let $\mom_\Delta$ denote the vector of the discrete moments obtained with the quadrature rule applied to $f$, namely
\[
\mom_\Delta = \ket{\phi(v) f}.
\]
Let $\disM$ be the discrete Maxwellian computed from the discrete moments $\mom_\Delta$. Then a simple semi-implicit time integration is
\begin{equation}\label{e:semiAP1}
f_{jk}^{n+1} = f_{jk}^{n} - \Dt D_x f_{jk}^n + \frac{\Dt}{\tau} \left( {\disM}_{jk}^{n+1} - f_{jk}^{n+1} \right),
\end{equation}
where $D_x$ is the discrete space derivative, for instance, the upwind derivative of \eqref{e:semidis}. The problem is how to compute $\disM$ which of course depends non linearly on $f$. To obtain $\disM$ at the new time step, we use conservation of mass, momentum and energy. Thus
\[
{\mom_{\Delta}}^{n+1}_j = {\mom_{\Delta}}_j^n -\left. \Dt D_x\right|_j \ket{v\,\phi(v) f^n},
\]
where the source term has disappeared, because $f$ and $\disM$ have exactly the same discrete moments. Thus, we compute the updated discrete moments using an explicit integration of the macroscopic equations. Once we have the new moments, we compute the discrete Maxwellian which reproduces the moments ${\mom_{\Delta}}^{n+1}_j$, at each grid point. This is $\disM^{n+1}= \Mop(\mom_{\Delta}^{n+1})$. A simple substitution in \eqref{e:semiAP1} gives
\begin{equation}
f_{jk}^{n+1} = \frac{\tau}{\tau + \Dt}\left(f_{jk}^{n} - \Dt D_x f_{jk}^n  \right)+ \frac{\Dt}{\tau+ \Dt}  {\disM}_{jk}^{n+1} .
\end{equation}
Note that as $\tau \to 0$, the update of the distribution function reduces to $f_{jk}^{n+1} = {\disM}_{jk}^{n+1}$. Then for the {\em new} time step, the equation for the moments with the upwind space derivative becomes
\[
{\mom_{\Delta}}^{n+2}_j = {\mom_{\Delta}}_j^{n+1} -\frac{\Dt}{\Dx} \left( F_{j+1/2}^{n+1}\ - F_{j-1/2}^{n+1} \right), 
\]
where the numerical flux $F_{j+1/2}$ at time $t^{n+1}$ is given by
\begin{equation}
F_{j+1/2}^{n+1} = \sum_k v_k^+ \phi(v_k) \disM(x_j,v_k,t^{n+1}) + \sum_k v_k^- \phi(v_k)\disM(x_{j+1},v_k,t^{n+1}),
\end{equation}
which is the kinetic flux vector splitting of classical gas dynamics.
This shows that the scheme is indeed AP, provided one starts the first time step with the data already in equilibrium, or close to equilibrium. More precisely, we say that the initial data are {\em well prepared} if
\begin{equation}\label{e:WellPrepared}
f(x_j,v_k,t=0) = \disM(x_j,v_k,t=0)+O(\tau).
\end{equation}

This approach started, as far as I know, in \cite{CoronPerth}, and was extended to arbitrary high order accuracy in \cite{PieracciniPuppo2007}. Then it was extended and improved by several authors. See the reviews \cite{DiMarcoPareschiActa} for more refined couplings between the explicit and the implicit integration schemes, and \cite{Jin2012review} for applications to different models.

\subsubsection*{AP schemes for mixtures of gases}

How does this technique applies to mixtures? The macroscopic equations for a two component mixture give
\begin{align}
& \partial_t \rho_1 + \nabla \cdot (\rho_1 u_1) = 0 \label{e:MassCons1}\\
& \partial_t \rho_2 + \nabla \cdot (\rho_2 u_2) = 0 \nonumber\\
& \partial_t (\rho_1 u_1) + \nabla \cdot (\rho_1 u_1 \otimes u_1 + \mathbb{P}_1) = m_1 \nu_{1,2}n_2n_1\left( u_{1,2}-u_1\right) \nonumber\\
& \partial_t (\rho_2 u_2) + \nabla \cdot (\rho_2 u_2 \otimes u_2 + \mathbb{P}_2) = m_2 \nu_{2,1}n_2n_1\left( u_{2,1}-u_2\right) \nonumber\\
& \partial_t E_1 + \nabla \cdot ( u_1 E_1 + \mathbb{P}_1u_1 + q_1) = m_1 \nu_{1,2}n_1n_2 
\left( u_{1,2}^2 -u_1^2 + \frac{d}{m_1} (T_{1,2} - T_1) \right)
 \nonumber \\
& \partial_t E_2 + \nabla \cdot ( u_2 E_2 + \mathbb{P}_2u_2 + q_2) = m_2 \nu_{2,1}n_1n_2 
\left( u_{2,1}^2 -u_2^2 + \frac{d}{m_2} (T_{2,1} - T_2) \right)
 \nonumber. 
\end{align}
This time the macroscopic equations have source terms, which are due to the macroscopic exchanges of momentum and heat between the two species. Moreover, these terms are stiff when the collision frequencies are high, thus they should be computed implicitely. On the other hand, we know the expressions of $u_{i,j}$ and $T_{i,j}$ due to the assumptions on the mixture velocities and temperatures and due to the  conservation of total momentum and energy. So, these quantities are known functions of the macroscopic speeds $u_i, i=1,2$ and temperature, $T_i, i=1,2$.

Thus we can modify the AP scheme for the single species, and obtain an AP scheme for the mixure. From $f_i^n(x,v)$, we compute the discrete moments $({\mom_i}_{\Delta})^{n}_j$ for each of the two species. Further, we also compute the pressure tensors and the heat fluxes at time $t^n$. We have enough information to update the transport terms appearing in the left hand side of system \eqref{e:MassCons1}. From the first two equations of the macroscopic system, we find $\rho_i^{n+1}, i=1,2$ and $n_i^{n+1}$ at each grid node. Substituting these updated quantities in the two momentum equations, we have
\begin{align}
& (\rho_1 u_1)^{n+1} =( \dots)^{n} + \Dt \, \nu_{1,2}\rho_1^{n+1}n_2^{n+1}\left( u_{1,2}^{n+1}-u_1^{n+1}\right) \nonumber\\
& (\rho_2 u_2)^{n+1} =( \dots)^{n} + \Dt \, \nu_{2,1}\rho_2^{n+1}n_1^{n+1}\left( u_{2,1}^{n+1}-u_2^{n+1}\right). \nonumber
\end{align}
Applying the known algebraic relations between $u_1, u_2$ and the two mixture velocities $u_{1,2}$ and $u_{2,1}$, we obtain a two by two linear algebraic system, which must be solved at each space grid point, yielding all velocities at the new time level $t^{n+1}$. We substitute these quantities in the energy equations, finding all updated temperatures, again solving a $2 \times 2$ linear algebraic system. 

Finally, since all macroscopic quantities are known, we compute all the discrete Maxwellians ${\disM}_i$ and ${\disM}_{i,j}$ at time $t^{n+1}$. In this fashion, we can compute the analog of \eqref{e:semiAP1} for the two distribution functions, with the source term evaluated at the new time. Thus again, we obtain a numerical method which is implicit in the stiff relaxation terms, and explicit in the transport step.

A similar technique can be applied also to the case of the single collision AAP mixture model, to yield an AP scheme for the update of the distribution functions also in the AAP case.

\subsubsection*{AP schemes for the ES-BGK model} 

The success of the AP time integration for the BGK model is based on the fact that the equilibrium distribution can be computed at the new time step without the need to solve implicitly a stiff system. In fact, when the macroscopic equations are updated, the stiff terms cancel because of the conservation of mass, momentum and energy. The situation is slightly more complicated for the ES-BGK model, because one of the macroscopic quantities needed to write the equilibrium distribution $\G$ is not conserved. However, it is still possible to update $\G$ in a very efficient way, see \cite{FilbetJin2010ES} and \cite{AlaiaThesis}.

Let $\disG$ be the discrete Gaussian computed from the discrete moments $\mom_\Delta$ and the discrete stress tensor $\mathbb{T}_\Delta$, i.e. the tensor $\mathbb{T}$ approximated with a quadrature rule in velocity space. The procedure is very similar to what we already presented for the case of the discrete Maxwellian $\disM$. A semi-implicit time integration for the ES-BGK model is
\begin{equation}\label{e:semiAP_ES}
f_{jk}^{n+1} = f_{jk}^{n} - \Dt D_x f_{jk}^n + \frac{\Dt}{\tau} \left( {\disG}_{jk}^{n+1} - f_{jk}^{n+1} \right),
\end{equation}
where $D_x$ is as before the discrete space derivative, for instance, the upwind derivative of \eqref{e:semidis}. Again, we need to compute $\disG$ which is linked non linearly to $f$. To obtain $\disG$ at the new time step, we use conservation of mass, momentum and energy. As in standard BGK, computing the discrete moments of the evolution equation, the source term disappears, because $f$ and $\disG$ have the same discrete moments, 
\[
{\mom_{\Delta}}^{n+1}_j = {\mom_{\Delta}}_j^n -\left. \Dt D_x \right|_j\ket{v\,\phi(v) f^n},
\]
Thus, $n_{\Delta}$, $u_{\Delta}$ and the temperature $T_{\Delta}$ are updated using an explicit integration of the macroscopic equations. This however is not enough to obtain $\disG^{n+1}$. We need also the evolution equation for the stress tensor \eqref{e:StressDecay}, which we now apply to the space inhomogeneous case. Integrating this equation backward in time, we have
\[
{\Theta_{\Delta}}_j^{n+1} = {\Theta_{\Delta}}_j^n - \left. \Dt D_x \right|_j\ket{v(v-u^n_{\Delta})\otimes(v-u^n_{\Delta})f^n} + \Delta t\frac{1-\omega}{\omega}\big( {T_{\Delta}}_j^{n+1}- {\Theta_{\Delta}}_j^{n+1}\big).
\]
Since $T^{n+1}_{\Delta}$ can be recovered from $\mom_{\Delta}^{n+1}$, $\Theta^{n+1}_{\Delta}$ can be easily obtained from previously computed quantities, and, from these, $\mathbb{T}^{n+1}_{\Delta}$. Now, $\mom_{\Delta}^{n+1}$ and $\mathbb{T}^{n+1}_{\Delta}$, completely define $\disG^{n+1}$, so that $f^{n+1}$ can be easily obtained from \eqref{e:semiAP_ES}.
It is easy to prove that this algorithm is indeed AP, see \cite{FilbetJin2010ES}.

\subsubsection*{AP schemes for the polyatomic BGK model} 

Finally, we note that an analogous approach can also be applied to the polyatomic model \eqref{eq:BGK_poly}. Here one computes the discrete moments integrating \eqref{eq:BGK_poly} on the whole velocity space $\R^D$. In particular, energy conservation gives,
\[
{E_{\Delta}}_j^{n+1} = {E_{\Delta}}_j^{n} - \left. \Dt D_x \right|_j\ket{\tfrac12 m v ||v||^2)f^n}.
\]
Since momentum and density are already known thanks to mass and momentum conservation, we also know the internal energy, $(\rho e)^{n+1} = E^{n+1}-\tfrac12 \rho^{n+1} (u^{n+1})^2 = n^{n+1}D\Teq^{n+1}$.

In this way the discrete moments $n_{\Delta}$, $u_{\Delta}$ and the equilibrium temperature $(\Teq)_{\Delta}$ at the time $t^{n+1}$ are found. Next, the equation for the evolution of the temperature \eqref{eq:ConsTrot} is discretized backward in time, to give
\[
(n\Lambda_r)_\Delta^{n+1} = (n\Lambda_r)_\Delta^{n}-\Delta t D_x \big(n\,u\, \Lambda_r\big)^n_\Delta + \Delta t\frac{n_\Delta^{n+1}}{Z_r\tau_\Delta^{n+1}}\big( {\Teq}_\Delta^{n+1} - {\Lambda_r}_\Delta^{n+1}\big),
\]
where the index $j$ has been dropped for simplicity. This equation yields
 ${\Lambda_r}_\Delta^{n+1}$, which, through energy conservation, yields ${\Lambda_t}_\Delta^{n+1}$. These two quantities were the missing ingredients to compute $\disM^{n+1}$, from which we find $f^{n+1}$ from \eqref{e:semiAP1}.

\subsection{High order asymptotic preserving schemes}
\label{s:HighOrder}
So far, I have discussed the issues underlying the discretization of kinetic models of BGK type using as an example a first order numerical scheme. However, as I already pointed out, in the numerical integration of kinetic equations it is important to use high order schemes, because these allow to use relatively coarse grids, still obtaining accurate solutions. We are now moving on to describe high order accurate schemes for BGK equations for mixtures.

The numerical integration in velocity space needed to compute the discrete moments is usually performed with the trapezoid rule, even for high order schemes, assuming that the bounds in the discrete velocity space are large enough to ensure that $f$ and its derivatives are negligible at the boundary of the discrete velocity space. In that case, \eqref{eq:Euler_mclaurin} holds and the quadrature error in velocity is usually smaller than the discretization errors in space and time. 

\subsubsection{High order space discretization}
\label{s:HighOrderSpace}
To fix ideas, we consider again a one-dimensional problem, as in eq. \eqref{e:1dBGK}, for a single species BGK model. The extension to the other models is straightforward, because the main difficulty occurs in the transport term, which is same for all cases. 

Usually, high order schemes fro hyperbolic problems are constructed as finite volume methods. One introduces the space {\em cell averages} of the distribution function
\[
\overline{f}_{j,k}(t) = \frac{1}{\Dx}\int_{x_j-\tfrac12\Dx}^{x_j+\tfrac12\Dx} f(x,v_k,t)\; \D x.
\]
Integrating \eqref{e:1dBGK} on each space cell and dividing by $\Dx$, one finds the finite volume formulation of the BGK model
\[
\totder{t}{\overline{f}}_{j,k}(t) = - \frac{1}{\Dx}(v_k)_1\left( f(x_j+\tfrac12\Dx,v_k,t) - f(x_j-\tfrac12\Dx,v_k,t) \right) + \overline{\frac{1}{\tau}\left( {\M}_{j,k}-f_{j,k} \right)},
\]
which indicates that it is necessary to compute the cell average of the whole source term.
This formulation is quite convenient for the evaluation of the space derivative of $f$, but complicates the source term. In fact, the cell average of a function coincides with its point value at the cell center only up to terms of order $O(\Dx)^2$. For higher order accuracy, the cell averages must be evaluated by quadrature, and this causes an unwanted coupling between the space cells.
In fact,  now a knowledge of the function being integrated at several quadrature nodes within the cell is needed. Thus, from the cell averages of $f$, it is necessary to reconstruct the point values of $f$, typically with the aid of a piecewise polynomial interpolator. But these algorithms require a stencil which couples together information coming from neighboring space cells. Since the algorithm is implicit in time, one must solve a non linear coupled algebraic system of equations, instead of a linear scalar equation for each cell. For this reason, it is preferable to use {\em finite difference schemes}, in which the main variables are the point values of the unknown function. Another possibility is to apply the treatment of stiff sources proposed in \cite{BoscarinoSource2018}.

In finite differences, the semidiscrete system of equations can be written as
\[
\totder{t}{{f}}_{j,k}(t) = - \left( \hat{F}_{\jp}((v_k)_1,f(v_k,t)) - \hat{F}_{\jm}((v_k)_1,f(v_k,t)) \right) + \frac{1}{\tau}\left( {\M}_{j,k}-f_{j,k} \right),
\]
where $\hat{F}_{\jp}$ is a numerical flux, see \cite{Shu97}, and \cite{PieracciniPuppo2007} for the application of finite differences to the BGK setting.
Note that in this case, to preserve high accuracy the grid spacing $\Dx$ must be uniform or at least smoothly varying in $x$. Suppose that the values of the flux function $F(f)$ are given at the grid points $x_j$ . In the present case, the flux function is linear, with $F(f)=vf$. The first step is to look for a function $\hat{F}$ that interpolates the data $F_j =F(f(.,x_j,.))$ in the sense of cell averages, namely
\[
F_j(v_k,t) = \frac{1}{\Dx} \int_{x_j-\tfrac12\Dx}^{x_j+\tfrac12\Dx} \hat{F}(x,v_k,t)\; \D x.
\]
Then the derivative of $F$ will be given by
\[
\partial_x \left. F(\cdot,v_k,t)\right|_{x_j} = \frac{1}{\Dx} \left( \hat{F}(x_{\jp},v_k,t) - \hat{F}(x_{\jm},v_k,t)\right). 
\]
We outline the main steps of the construction: more details will be found in \cite{PieracciniPuppo2007}. The approximation to $\hat{F}$ is typically a piecewise polynomial function, with jump discontinuities at the cell borders $x_{j\pm\tfrac12}$. To ensure stability, as in the first order case, it is necessary to pick information coming from the correct direction. In other words, it is necessary to introduce upwinding. In particular, we will use flux splitting, which is particularly straightforward for the convective term of the BGK model.
Thus, we write the flux $F$ as the sum of its positive and negative parts: $F = F^++F^-$, where $F^+$ and $F^-$ have only non-negative (respectively, non-positive) eigenvalues. In our case, the flux splitting will depend on the sign of $v$, namely:
\[
F^+=\begin{cases} vf & \mbox{if}\; v \geq 0 \\ 0 & \mbox{otherwise}\end{cases}
\qquad 
F^-=\begin{cases} 0 & \mbox{if}\; v \geq 0 \\ vf & \mbox{otherwise}\end{cases}
\]
Next, two reconstructions are computed, one from $F^+$ and one from $F^-$, which will yield respectively $\hat{F}^+(x,v,t)$ and $\hat{F}^-(x,v,t)$. Both of them are piecewise polynomial functions in $x$ with jumps at the cell interfaces $x_{j\pm 1/2}$. To enforce upwinding, we pick the value from the left
for the positive flux, $\hat{F}^+ =\hat{F}^+(x^-_{\jp})$, and we pick the value from the
right for the negative flux: $\hat{F}^- =\hat{F}^-(x^+_{\jp})$. Thus the numerical flux at the cell interface $x_{j+1/2}$ is given by
\begin{equation}
\hat{F}_{\jp}= \hat{F}^+(x_{\jp}^-)+\hat{F}^-(x_{\jp}^+)
\end{equation}
and the conservative approximation to the space derivative will be
given by:
\begin{equation}
\partial_x (v_1 \, f)|_{j} = \frac{1}{\Delta x} \left( \hat{F}_{j+1/2} - \hat{F}_{j-1/2}
\right).  \label{eq:cons_fd}
\end{equation}
In the present case, the structure of $F$ is particularly simple,
and one has: \begin{equation} \hat{F}_{j+1/2}(v) = \max(v,0)
\hat{f}(x_{j+1/2}^-) + \min(v,0) \hat{f}(x_{j+1/2}^+),
\end{equation}
where $\hat{f}$ is a piecewise polynomial function such that $f_j
= \frac{1}{\Delta x} \int_{I_j} \hat{f} \, dx$.

For example, for the first order scheme $\hat{f}$ is piecewise
constant, namely $\hat{f}(x)|_{I_j}\equiv f_j$. Thus
$\hat{f}(x_{j+1/2}^-) = f_j$ while $\hat{f}(x_{j+1/2}^+)=f_{j+1}$
and the numerical flux in this case will be given by:
\begin{equation}
\hat{F}_{j+1/2}(v) = \max(v,0) f_j + \min(v,0) f_{j+1}.
\label{eq:first_flux}
\end{equation}
The reconstruction is carried out with piecewise  non oscillatory polynomials. There are several algorithms in the literature. A thorough review of WENO and ENO reconstructions can be found in \cite{Shu97}. A more recent contribution is \cite{CWENO2018}.

\subsubsection{High order time discretization}
\label{s:HighOrderTime}

The purpose of this section is to increase the time accuracy, while preserving a splitting between the fast source term, integrated implicitly, and the slow transport part, to be integrated explicitly. We will consider additive Runge-Kutta schemes \cite{Carpenter2003}, applied to kinetic problems, see \cite{PareschiRusso2005} or \cite{DiMarcoPareschi2013}.

 To set the notation for additive Runge-Kutta IMEX schemes, we
consider the autonomous ODE problem:
\begin{equation}
\begin{array}{l} y'(t) = f(y) + \frac{1}{\tau}g(y)  \\ y(t_0) = y_0. \end{array}
\label{eq:model_ode}
\end{equation}
We suppose that $0 < \tau << 1$, i.e. $g/\tau$ is stiff, so that we wish to integrate it
implicitly, while $f$ is non stiff, but highly non linear, which
means that  for $f$ an explicit scheme is more efficient.

Let $\tilde A = (\tilde a_{is})$ and $A = (a_{is})$ be two $\nu
\times \nu$ matrices, with $\tilde A$ strictly lower triangular, and
$A$ lower triangular, with non zero terms along the diagonal, and let $\tilde b$, $b$, $\tilde c$, $c$ be 
coefficient vectors with $\nu$ elements. In other words, for the implicit part we consider a DIRK scheme (diagonally implicit Runge-Kutta). More effectively, an IMEX Runge-Kutta
scheme is represented by the following double Butcher's {\em
tableaux}:
\bigskip

\begin{minipage}{0.45\textwidth}
 \vspace{0pt}
\begin{center}
\begin{tabular}{l|r}
$\tilde{c}$ & {\Large $\tilde{A}$} \\
\hline
 & $\tilde{b}^T$
\end{tabular}
\end{center}
\end{minipage}
\begin{minipage}{0.45\textwidth}
 \vspace{0pt}
\begin{center}
\begin{tabular}{l|r}
$c$ & {\Large $A$} \\
\hline
 & $b^T$
\end{tabular}
\end{center}
\end{minipage}

\bigskip

\noindent where $c$ and $\tilde{c}$ are included for completeness, although they are not needed for the autonomous system we are considering, and the tableau on the right refers to the implicit part. 
The resulting numerical scheme for \eqref{eq:model_ode} is:
\begin{eqnarray}
y^{n+1} & = & y^n + \Dt \sum_{i=1}^\nu \tilde b_i f(y^{(i)}) +
\frac{\Dt}{\tau}  \sum_{i=1}^\nu b_i  g(y^{(i)})   \label{eq:imex},
\end{eqnarray}
where the stage values $y^{(i)}$ are given by:
\begin{eqnarray}
y^{(1)} & = & y^n + \frac{\Dt}{\tau} a_{11} \, g(y^{(1)}) \\
y^{(i)}  & = & y^n + \Dt \sum_{l=1}^{i-1} \tilde a_{il} f(y^{(l)})
+ \frac{\Dt}{\tau} \sum_{l=1}^i a_{il} \, g(y^{(l)}). \label{eq:imex_stage}
\end{eqnarray}

The coefficients of the Butcher's tableaux are computed in order
to maximize accuracy. Further, the implicit scheme must be
L-stable, to ensure that the numerical solution relaxes on the
equilibrium solution, if $\tau$ is very small.
Moreover, it is desirable that the IMEX scheme becomes a high
order explicit numerical scheme for the conserved variables, when
$\tau \rightarrow 0$. 

The IMEX scheme is called of A-type in \cite{DiMarcoPareschi2013} if the matrix $A$ is invertible, i.e. $a_{ii}\neq 0,\; \forall i$.It is easy to see that the IMEX scheme is AP if the matrix $A$ is invertible, see \cite{PareschiRusso2005}. In fact, the equation for the $i$-th stage is
\[
f^{(i)}   =  \left[ f^n - \Dt \sum_{\ell=1}^{i-1} \tilde a_{i\ell} D_x(vf\stl)
    + \frac{\Dt}{\tau} \sum_{\ell=1}^{i-1} a_{i\ell} \, (\disM\stl-f\stl)\right]
    + \frac{\Dt}{\tau}  a_{ii} \, (\disM^{(i)}-f^{(i)}).
\]
The first parenthesis contains only previously computed quantities. The last term, when $\tau \to 0$, and $a_{ii}\neq 0$ yields $f^{(i)}\to \disM^{(i)}$. Thus, at each stage, $f^{(i)}$ is projected on the local discrete Maxwellian. Taking moments of the previous equation, in the limit $\tau \to 0$, one obtains
\[
U^{(i)}   =  U^n - \Dt \sum_{\ell=1}^{i-1}  \tilde{a}_{il} D_x \ket{v\phi \disM\stl}, \qquad \disM^{(i)}=\Mop(U^{(i)})
\]
which is the computation of the $i$-th stage for the solution of Euler equation with the explicit Runge-Kutta scheme with Butcher tableau $\tilde{A}$. I recall that $\Mop$ denotes the Maxwellian operator, which computes the discrete Maxwellian built solving the non linear algebraic system \eqref{e:DiscreteMaxwellian}, starting from the discrete moments $U^{(i)}$. Further, closing the Runge Kutta step one obtains
\begin{equation}\label{e:fUpdate}
f^{n+1}   =  f^n - \Dt \sum_{i=1}^\nu \tilde b_{i} D_x(vf^{(i)})
    + \frac{\Dt}{\tau} \sum_{i=1}^{\nu} b_{i} \, (\disM^{(i)}-f^{(i)}).
\end{equation}
Again, computing moments, in the limit $\tau \to 0$, the following equation is obtained
\[
U^{n+1}   =  U^n - \Dt \sum_{i=1}^{\nu}  \tilde{b}_i D_x \ket{v\phi \disM^{(i)}}. 
\]
In other words, as $\tau \to 0$, the IMEX scheme becomes the explicit scheme for the Euler equations, with a tableau which coincides with the tableau for the explicit part of the IMEX pair. However, a couple of remarks are in order.

The first stage of an A-type scheme applied to the BGK equation is
\[
f^{(1)}   =   f^n 
    + \frac{\Dt}{\tau}  a_{11} \, (\disM^{(1)}-f^{(1)}),
\]
with $\disM^{(1)}$ computed from the discrete moments $U^{(1)}$, which are given by $U^{(1)}=U^n$. Thus, the scheme is implicit, but it starts projecting $f$ on the old Maxwellian. In this fashion, even the first update of the moments (which occurs from the second stage onwards) is carried out at equilibrium, when $\tau\to 0$. For this reason, schemes of type A are AP even with initial data which are not well prepared. 

However, the first order scheme with which this discussion started, namely \eqref{e:semiAP1}, updates the moments {\em before} projecting $f$ on the Maxwellian, and still is AP, albeit on well prepared initial data. Applying the formalism just introduced, we see that it is actually a two stages IMEX scheme given by
\begin{align*}
f^{(1)} & =  f^n \\
\mom^{(1)}  & =  \mom^n \\
\mom^{(2)} & =  \mom^n - \Dt D_x\ket{v\phi f^{(1)}} \\
f^{(2)}  & =  f^n - \Dt D_x \big(v  f^{(1)}\big)
    + \frac{\Dt}{\tau} (\disM^{(2)}-f^{(2)}).
\end{align*}
Then the final update is given by
\begin{align*}
f^{n+1}  & =  f^n - \Dt D_x \big(v  f^{(1)}\big)
    + \frac{\Dt}{\tau} (\disM^{(2)}-f^{(2)}) \\
\mom^{n+1} & =  \mom^n - \Dt D_x\ket{v\phi f^{(1)}}
\end{align*}
which coincides with \eqref{e:semiAP1}, because $f^{(1)}=f^n$. Note that the final update coincides with the last stage, i.e. $f^{n+1}=f^{(2)}$ and $\mom^{n+1} = \mom^{(2)}$.
This corresponds to the IMEX scheme described by the following tableaux
\bigskip

\begin{minipage}{0.45\textwidth}
 \vspace{0pt}
\begin{center}
\begin{tabular}{l|r}
$\tilde{c}$ & {$\begin{array}{lr}0 & 0 \\ 1 & 0 \end{array}$} \\
\hline
 & $\begin{array}{lr}1 & 0  \end{array}$
\end{tabular}
\end{center}
\end{minipage}
\begin{minipage}{0.45\textwidth}
 \vspace{0pt}
\begin{center}
\begin{tabular}{l|r}
$c$ & {$\begin{array}{lr}0 & 0 \\ 0 & 1 \end{array}$} \\
\hline
 & $\begin{array}{lr} 0 & 1  \end{array}$
\end{tabular}
\end{center}
\end{minipage}

\bigskip

We see therefore that the scheme I proposed initially is not of type A, because the implicit matrix $A$ is not invertible. However, we have already seen that the resulting scheme is nevertheless AP, provided the initial data are well prepared. The scheme above belongs to a class of IMEX schemes proposed in \cite{Ascher1997} and \cite{Carpenter2003} which are also used in the context of kinetic equations in \cite{DiMarcoPareschi2013}, where they are called schemes of type CK. These schemes are given by tableaux with the following form

\bigskip

\begin{minipage}{0.45\textwidth}
 \vspace{0pt}
\begin{center}
\begin{tabular}{l|r}
$\tilde{c}$ & {$\begin{array}{lr}0 & 0 \\ \tilde{a} & \tilde{A} \end{array}$} \\
\hline
 & $\tilde{b}^T$
\end{tabular}
\end{center}
\end{minipage}
\begin{minipage}{0.45\textwidth}
 \vspace{0pt}
\begin{center}
\begin{tabular}{l|r}
$c$ & {$\begin{array}{lr}0 & 0 \\ a & A \end{array}$} \\
\hline
 & $b^T$
\end{tabular}
\end{center}
\end{minipage}

\bigskip
\noindent where now $\tilde{A}$ and $A$ are $(\nu-1) \times (\nu-1)$ lower triangular matrices, with $A$ invertible, and $\tilde{A}$ strictly lower triangular, while $\tilde{a}$ and $a$ are $\nu-1$ components vectors.
Schemes of type CK are also AP, but they need well prepared initial data.

A second remark concerns the final update of $f$ occurring in \eqref{e:fUpdate}. If at all stages $f$ has been projected on equilibrium, i.e. $f^{(i)}=\disM^{(i)}$, then the effect of \eqref{e:fUpdate} is to drive $f$ away from equilibrium. This does not occur in \eqref{e:semiAP1} because the final update actually coincides with the last stage computed. This property is satisfied if $b_1=a_{\nu-1}, b_{i+1}=A_{\nu-1,i}, i=1,\dots,\nu-1$ and $\tilde{b}_1=\tilde{a}_{\nu-1}, \tilde{b}_{i+1}=\tilde{A}_{\nu-1,i}, i=1,\dots,\nu-1$. In short, for these schemes the vectors $b$ and $\tilde{b}$ coincide with the last row of the matrices in their respective tableau. The schemes that satisfy this property are called GSA (Globally Stiffly Accurate) in \cite{DiMarcoPareschi2013}, and they are particularly suited for applications in kinetic problems, because they are AP not only on the evolution of the moments, but also on the evolution of $f$. Note that the GSA property introduces $2\nu$ constraints on the coefficients of the two Butcher tableaux composing the IMEX pair. It is not surprising therefore that the order conditions of GSA schemes are penalizing: one needs four stages to obtain a second order accurate GSA A type scheme, see \cite{Ascher1997, Carpenter2003} for more details.

\bibliographystyle{abbrv}
\bibliography{CocktailBib}{} 

\end{document}